\documentclass[11pt,a4paper]{article}
\usepackage[left=25mm,top=25mm,right=25mm]{geometry}

\usepackage[latin1]{inputenc}

\newcommand{\notetoself}[1]{\par\noindent\textsc{\qquad** #1 **}\par\noindent}

\usepackage{amsmath}
\usepackage{amsfonts}
\usepackage{amssymb}
\usepackage{amsthm}
\usepackage{enumerate}
\usepackage[mathscr]{euscript}
\usepackage{url}
\usepackage{comment}

\newtheorem{theorem}{Theorem}
\newtheorem{lemma}[theorem]{Lemma}
\newtheorem{proposition}[theorem]{Proposition}
\newtheorem{corollary}[theorem]{Corollary}
\newtheorem{observation}[theorem]{Observation}

\theoremstyle{definition}
\newtheorem{definition}[theorem]{Definition}
\newtheorem{property}[theorem]{Property}

\theoremstyle{remark}
\newtheorem{remark}[theorem]{Remark}
\newtheorem{example}[theorem]{Example}

\newcommand{\omittext}[1]{}

\newcommand{\N}{\mathbb{N}}

\newcommand{\Q}{\mathbb{Q}}
\newcommand{\R}{\mathbb{R}}

\newcommand{\I}{\mathbb{I}}

\newcommand{\Hl}{\mathbb{H}}

\renewcommand{\P}{\mathcal{P}}
\renewcommand{\Pr}{\mathbb{P}}
\newcommand{\Ex}{\mathbb{E}}

\newcommand{\abs}{\mathrm{abs}}

\newcommand{\diam}{\mathrm{diam}}
\newcommand{\dom}{\mathrm{dom}}

\newcommand{\rng}{\mathrm{rng}}
\newcommand{\supp}{\mathrm{supp}}

\newcommand{\fto}{\rightarrow}
\newcommand{\mvfto}{\rightrightarrows}
\newcommand{\pfto}{\rightharpoonup}
\newcommand{\mfto}{\rightsquigarrow}
\newcommand{\psfto}{\rightharpoonup}

\newcommand{\lmsto}{\rightsquigarrow}

\newcommand{\clI}{\overline{I}}

\newcommand{\clN}{\,\overline{\!N}}
\newcommand{\clB}{\overline{B}}

\newcommand{\calB}{\mathcal{B}}

\newcommand{\calT}{\mathcal{T}}
\newcommand{\calV}{\mathcal{V}}
\newcommand{\calX}{\mathcal{X}}
\newcommand{\calY}{\mathcal{Y}}

\newcommand{\opset}{\mathcal{O}}
\newcommand{\clset}{\mathcal{A}}

\newcommand{\ctsfn}{\mathcal{C}}

\newcommand{\meas}{\mathcal{M}}

\newcommand{\standardtype}[1]{\mathbb{#1}}
\newcommand{\standardfunctor}[1]{\mathcal{#1}}
\newcommand{\usertype}[1]{{#1}}
\renewcommand{\usertype}[1]{\mathbb{#1}}

\newcommand{\tpBl}{\standardtype{B}}
\newcommand{\tpSi}{\standardtype{S}}
\newcommand{\tpNat}{\standardtype{N}}
\newcommand{\tpHl}{\standardtype{H}_<}
\renewcommand{\tpHl}{\standardtype{H}}
\newcommand{\tpIv}{\standardtype{I}}
\newcommand{\tpRe}{\standardtype{R}}
\newcommand{\tpMe}{\standardfunctor{M}}

\newcommand{\tpPr}{\standardfunctor{P}}
\newcommand{\tpCts}{\standardfunctor{C}}

\newcommand{\tpRV}{\standardfunctor{R}}
\newcommand{\tpLM}{\standardfunctor{M_<}}
\newcommand{\tpOp}{\standardfunctor{O}}
\newcommand{\tpCl}{\standardfunctor{A}}
\newcommand{\tpR}{\usertype{R}}
\newcommand{\tpT}{\usertype{T}}
\newcommand{\tpW}{\usertype{W}}
\newcommand{\tpX}{\usertype{X}}
\newcommand{\tpY}{\usertype{Y}}
\newcommand{\tpZ}{\usertype{Z}}
\newcommand{\tpOmega}{\usertype{\Omega}}
\renewcommand{\tpOmega}{\Omega}

\newcommand{\chif}{\raisebox{0.44ex}{\ensuremath{\chi}}}

\newcommand{\true}{{\sf{T}}}
\newcommand{\false}{{\sf{F}}}
\newcommand{\indet}{\bot}

\begin{document}

\title{Computable Random Variables and Conditioning}
\author{Pieter Collins\\Department of Data Science and Knowledge Engineering\\Maastricht University\\\texttt{pieter.collins@maastrichtuniversity.nl}}
\date{22 December 2020}
\maketitle

\begin{abstract}
The aim of this paper is to present an elementary computable theory of random variables, based on the approach to probability via valuations.
The theory is based on a type of lower-measurable sets, which are controlled limits of open sets, and extends existing work in this area by providing a computable theory of conditional random variables.
The theory is based within the framework of type-two effectivity, so has an explicit direct link with Turing computation, and is expressed in a system of computable types and operations, so has a clean mathematical description.
\end{abstract}




\section{Introduction}

In this paper, we present a computable theory of probability and random variables.
The theory is powerful enough to provide a theoretical foundation for the rigorous numerical analysis of discrete-time continuous-state Markov chains and stochastic differential equations~\cite{Collins2014ARXIV}.
We provide an exposition of the approach to probability distributions using valuations and the development of integrals of positive lower-semicontinuous and of bounded continuous functions, and on the approach to random variables as limits of almost-everywhere defined continuous partial functions.


We use the framework of \emph{type-two effectivity (TTE)}~\cite{Weihrauch1999}, in which computations are performed by Turing machines working on infinite sequences, as a foundational theory of computability.
We believe that this framework is conceptually simpler for non-specialists than the alternative of using a domain-theoretic framework.
Since in TTE we work entirely in the class of quotients of countably-based (QCB) spaces, which form a cartesian closed category, many of the basic operations can be carried out using simple type-theoretic constructions such as the $\lambda$-calculus.

In this paper, we deal with computability theory, rather than constructive mathematics.
In practice, this means that we allow recursive constructions, and so accept the axiom of dependent (countable) choice, but since not all operations are decidable, we do not accept the law of the excluded middle.
However, proofs of correctness of computable operators may use non-computable functions and proof by contradiction.

We assume that the reader has a basic familiarity with classical probability theory~(see e.g.~\cite{Shiryaev1995,Pollard2002}.
Much of this article is concerned with giving computational meaning to classical concepts and arguments.
The main difficulty lies in the use of $\sigma$-algebras in classical probability, which have poor computability properties due to the presence of countable unions and complementation.
Instead, we use only topological constructions, which can usually be effectivised directly.
We can compute lower bounds (but not upper bounds) to the measure of open sets, and extend results to measurable sets using completion constructions.
Similarly we define types of measurable and integral functions as completions of types of (piecewise) continuous functions.


In Section~\ref{sec:computableanalysis}, we briefly introduce the foundations of computable analysis.
In Section~\ref{sec:valuations}, we describe the approach to probability theory using valuations.
The main results are in Section~\ref{sec:randomvariables}, in which we give a complete theory of random variables in separable metric spaces.
We begin by constructing types of measurable sets and measurable functions on a given base probability space $(\Omega,P)$ using completion operators, similarly to existing approaches in the literature.
We show that the distribution of a random variable is computable, and conversely, that for any valuation we can construct a realisation by a random variable, similarly to results of~\cite{SchroederSimpson2006JCMPLX,HoyrupRojas2009INFCOMP}.
We then show the trivial result that the product of two random variables is computable, and the classical result~\cite{MannWald1943,BishopCheng1972} that the image of a random variable under a continuous function is computable.
We define the expectation of a random variable, and types of integrable random variables in the standard way.
Finally, we discuss conditioning of random variables, and show how a random variable can be computed from its conditional expectation.

\subsection*{Comparison with other approaches}

An early fairly complete constructive theory of measure theory based on the Daniell integral was developed in~\cite{BishopCheng1972} was presented in~\cite{BishopBridges1985} and~\cite{Chan1974}.
The theory is developed using abstract \emph{integration spaces}, which are triples $(X,L,I)$ where $X$ is a space, $L$ a subset of test functions $X\fto\R$ and $I:L\fto \R$ satifying properties of an integral.
The integral is extended from test functions to integrable functions by taking limits.
Measurable functions are those which can be uniformly approximated by integrable functions on large-measure sets.
It is shown that the image of a measurable function under a continuous function is measurable, and analogue of our Theorem~\ref{thm:imagerandomvariable}.
Measurable sets are defined via \emph{complemented sets}, which are pairs of sets $(A_1,A_0)$ such that $A_0\cap A_1=\emptyset$, and are measurable if $A_0\cup A_1$ is a full set.
Abstract measure spaces are defined in terms of measurable sets, and shown to be equivalent to integration spaces.

A standard approach to a constructive theory of probability measures, as developed in~\cite{JonesPlotkin1989,Edalat1995DTI,SchroederSimpson2006ENTCS,Escardo2009}, is through \emph{valuations}, which are essentially measures restricted to open sets.
Explicit representations of valuations within the framework of type-two effectivity were given in~\cite{Schroeder2007}.
Valuations satisfy the \emph{modularity} property $\nu(U_1)+\nu(U_2) = \nu(U_1\cup U_2) + \nu(U_1\cap U_2)$, and (monotonic) continuity $\nu(U_\infty)=\lim_{n\to\infty} \nu(U_n)$ whenever $U_n$ is an increasing sequence of open sets with $U_\infty=\bigcup_{n\in\N}U_n$.
Relationships between valuations and Borel measures were given in ~\cite{Edalat1995DTI} and extended in~\cite{AlvarezManilla2002} and~\cite{GoubaultLarrecq2005}.

The most straightforward approach to integration is the \emph{Choquet} or \emph{horizontal} integral, a lower integral introduced within the framework of domain theory  in~\cite{Tix1995}; see also~\cite{Konig1995,Lawson2004}.
The lower integral on valuations in the form we use was given in~\cite{Vickers2008}.
The monadic properties of the lower integral on valuations, which has type $(\tpX\fto\tpR^+_<)\fto\tpR^+_<$ were noted by~\cite{Vickers2011PP}.
A similar monadic approach to probability measures was used in~\cite{Escardo2009} to develop a language EPCL for nondeterministic and probabilistic computation.
Here, the type of probability measures on the Cantor space $\Omega=\{0,1\}^\omega$ was identified with the type of integrals $(\Omega \fto \I) \fto \I$ where $\I=[0,1]$ is the unit interval.


An alternative to the use of valuations is that of~\cite{CoquandPalmgren2002AML}.
The exposition is given in terms of general \emph{boolean rings}, but in the language of sets, a measure $\mu$ is given satisfying the modularity condition, and extended by the completion under the metric $d(S_1,S_2)=\mu(S_1 \triangle S_2)$ where $S_1 \triangle S_2 = (S_1 \setminus S_2) \cup (S_2\setminus S_1)$.
In~\cite{WuDing2005MLQ,WuDing2006AML} a concept of \emph{computable measure space} with a concrete representation was given using a ring of subsets $R$ generating the Borel $\sigma$-algebra.
Disadvantages of this approach are that the elements of $R$ must have a computable measure, introducing an undesirable dependency between the measure and the ``basic'' sets which is not present in the approach using valuations.

In the approach presented here, we start with the use of valuations, since these are intrinsic given a base type $\tpX$.
For a fixed valuation, we can extend to a class of \emph{lower-measurable sets}, and also give a definition of measurable set using a completion on complemented open sets (equivalently, on topologically regular sets).
We do not use integration spaces, since we feel that the concept of measure is more fundamental than that of integral.


In the approach of~\cite{Spitters2006}, integrable functions are defined as limits of simple functions with respect to the measurable sets of~\cite{CoquandPalmgren2002AML}.
Measurable function are defined as limits of effectively-converging Cauchy sequences with respect to the pseudometric $d_h(f,g):= \int |f-g| \wedge h$ for positive integrable $h$ and integrable $f,g$.
This work was generalised to Riesz spaces in~\cite{CoquandSpitters2009}.

Random variables over discrete domains were defined in~\cite{Mislove2007}, based on work of~\cite{Varacca2002}, and extended to continuous domains in~\cite{GoubaultLarrecqVaracca2011}.
A continuous random variable in $\tpX$ was defined as a pair $(\nu,f)$ where $\nu$ is a continuous valuation on $\Omega=\{0,1\}^\omega$, and $f$ is a continuous map from $\supp(\nu)$ to $\tpX$, where $\supp(\nu)$ is the smallest closed set $A$ such that $\nu(A)=1$.
A difficulty with this construction is that different random variables require different valuations on the bases space $\Omega$, which makes computation of joint distributions problematic.

In this paper, we define measurable functions as those for which the preimage of an open set is a lower-measurable set and satisfy the natural properties.
This mimics the standard property that the preimage of an open set under a measurable function is (Borel) measurable.
Since measurable functions are in general uncomputable, we do not even attempt to define the ``image of a point''.

A similar approach to~\cite{Spitters2006} is also possible, defining random variables (measurable functions) directly by completion with respect to the Fan metric
\( d(X,Y) = \sup\!\big\{ \varepsilon\in\Q^+ \mid \ P\big(\{\omega\in\Omega \mid d(X(\omega),Y(\omega))>\varepsilon\}\big) > \varepsilon\big\} .\)
The resulting theory is essentially equivalent to that of~\cite{BishopCheng1972}, but developed in reverse.
The resulting representation is equivalent to that using lower-measurable sets.

In~\cite{SchroederSimpson2006JCMPLX}, an alternative representation of valuations and measures on $\tpX$ was developed by defining a valuation $\pi$ on (a subset of) the sequence space $\{0,1\}^\omega$, and pushing-forward by the representation $\delta$ of $\tpX$, yielding $\nu(U)=\pi(\delta^{-1}(U))$.
This representation of valuations is similar the valuation induced by our random variables, except that our random variables are obtained by taking limits, so we need to prove separately that the valuation induced by a random variable is computable.

It was further shown~\cite{SchroederSimpson2006JCMPLX} and that the alternative representations of valuations always exist on sufficiently nice spaces, which can be seen as a realisation result for valuations, where the representation $\delta$ of the space $\tpX$ is a random variable.
In~\cite{HoyrupRojas2009INFCOMP}, a theory of probability was developed for the study of algorithmic randomness, and a similar representation result for valuations  was given, here allowing both the base-space measure $\pi$ and point-representation $\delta$ to be given.
In this paper, we also show that valuations have concrete realisations by random variables, but our result constructs random variables relative to to uniform probability measure on the base space $\{0,1\}^\omega$, and a Cauchy sequence of (continuous) functions rather than a single function on a $G_\delta$-set.

The problem of finding conditional expectation, which classically uses the Radon-Nikodym derivative, was shown to be uncomputable by~\cite{HoyrupRojasWeihrauch2011}.
This means that computably, there is a difference between a random variable, and a ``conditional random variable''.
Here we show that given a conditional random variable $Y|\calX$, and an $\calX$-measurable random variable $X$, we can effectively compute $Y$.
This result is important for stochastic processes, in which we typically can compute the distribution of $X_{t+1}$ given $X_t=x_t$.

\section{Computable Analysis}
\label{sec:computableanalysis}

In the theory of type-two effectivity, computations are performed by Turing machines acting on \emph{sequences} over some alphabet $\Sigma$.
A computation performed by a machine $\meas$ is \emph{valid} on an input $p\in\Sigma^\omega$ if the computation does not halt, and writes infinitely many symbols to the output tape.
A type-two Turing machine therefore performs a computation of a partial function $\eta:\Sigma^\omega \pfto \Sigma^\omega$; we may also consider multi-tape machines computing $\eta:(\Sigma^\omega)^n \pfto (\Sigma^\omega)^m$.
It is straightforward to show that any machine-computable function $\Sigma^\omega \pfto \Sigma^\omega$ is continuous on its domain.

In order to relate Turing computation to functions on mathematical objects, we use \emph{representations} of the underlying sets, which are partial surjective functions $\delta:\Sigma^\omega \psfto \tpX$.
An operation $\tpX\fto \tpY$ is $(\delta_\tpX;\delta_\tpY)$-computable if there is a machine-computable function $\eta:\Sigma^\omega \pfto \Sigma^\omega$ with $\dom(\eta)\supset\dom(\delta_\tpX)$ such that $\delta_\tpY\circ \eta = f\circ\delta_\tpX$ on $\dom(\delta_\tpX)$.
Representations are equivalent if they induce the same computable functions.
A \emph{computable type} is a pair $(\tpX,[\delta])$ where $\tpX$ is a space and $[\delta]$ is an equivalence class of representations of $\tpX$.

If $X$ is a topological space, we say that a representation $\delta$ of $X$ is an \emph{admissible quotient representation} if (i) whenever $f:\tpX\fto \tpY$ is such that $f\circ\delta$ is continuous, then $f$ is continuous, and (ii) whenever $\phi:\Sigma^\omega\pfto \tpX$ is continuous, there exists continuous $\eta:\Sigma^\omega\pfto\Sigma^\omega$ such that $\phi=\delta\circ\eta$.
Any space with a quotient representation is a quotient of a subset of the countably-based space $\Sigma^\omega$, and is a \emph{sequential space}.
(A topological space is a \emph{sequential space} if any subset $W$ for which $x_n\to x_\infty$ with $x_\infty \in W$ implies $x_n\in W$ for all sufficiently large $n$, is an open set.)

A function $f:\tpX \fto \tpY$ is \emph{computable} if there is a machine-computable function $\eta:\Sigma^\omega \pfto \Sigma^\omega$ with $\dom(\eta)\supset\dom(\delta_\tpX)$ such that $\delta_\tpY\circ \eta = f\circ\delta_\tpX$ on $\dom(\delta_\tpX)$.
A multivalued function $F:\tpX \mvfto \tpY$ is \emph{computably selectable} if there is a machine-computable function $\eta:\Sigma^\omega \pfto \Sigma^\omega$ with $\dom(\eta)\supset\dom(\delta_\tpX)$ such that $\delta_\tpY\circ \eta \in F\circ\delta_\tpX$ on $\dom(\delta_\tpX)$; note that different names of $x\in\tpX$ may give rise to different values of $y \in \tpY$.

The category of computable types with (sequentially) continuous functions is Cartesian closed, and the computable functions yield a Cartesian closed subcategory.
For any types $\tpX$, $\tpY$ there exist a canonical product type $\tpX\times\tpY$ with computable projections $\pi_{\tpX}:\tpX\times\tpY\fto\tpX$ and $\pi_{\tpY}:\tpX\times\tpY\fto\tpY$, and a canonical exponential type $\tpY^\tpX$ such that evaluation $\epsilon:\tpY^\tpX\times\tpX \rightarrow \tpY:(f,x)\mapsto f(x)$ is computable.
Since objects of the exponential type are continuous function from $\tpX$ to $\tpY$, we also denote $\tpY^\tpX$ by $\tpX\fto\tpY$ or $\tpCts(\tpX;\tpY)$; in particular, whenever we write $f:\tpX\fto\tpY$, we imply that $f$ is continuous.
There is a canonical equivalence between $(\tpX \times \tpY) \fto \tpZ$ and $\tpX \fto (\tpY \fto \tpZ)$ given by $\tilde{f}(x):\tpY\fto\tpZ:\tilde{f}(x)(y)=f(x,y)$.

There are canonical types representing basic building blocks of mathematics, including the natural number type $\tpNat$ and the real number type $\tpRe$.
We use a three-valued logical type with elements $\{\false,\true,\indet\}$ representing \emph{false}, \emph{true}, and \emph{indeterminate} or \emph{unknowable}, and its subtypes the \emph{Boolean} type $\tpBl$ with elements $\{\false,\true\}$  and the \emph{Sierpinski} type $\tpSi$  with elements $\{\true,\indet\}$.
Given any type $\tpX$, we can identify the type $\tpOp(\tpX)$ of open subsets $U$ of $\tpX$ with $\tpX\fto\tpSi$ via the characteristic function $\chi_U$.
Further, standard operations on these types, such as arithmetic on real numbers, are computable.

A sequence $(x_n)$ is an \emph{effective Cauchy sequence} if $d(x_m,x_n) < \epsilon_{\max(m,n)}$ where $(\epsilon_n)_{n\in\N}$ is a known computable sequence with $\lim_{n\to\infty}\epsilon_n=0$, and a \emph{strong Cauchy sequency} if $\epsilon_n=2^{-n}$.
The limit of an effective Cauchy sequence of real number is computable.

We shall also need the type $\tpHl \equiv \tpRe^{+,\infty}_<$ of positive real numbers with infinity under the lower topology.
The topology on the lower halfline $\tpHl$ is the toplogy of lower convergence, with open sets $(a,\infty]$ for $a\in\R^+$ and $\tpHl$ itself.
A representation of $\tpHl$ then encodes an increasing sequence of positive rationals with the desired limit.
We note that the operators $+$ and $\times$ are computable on $\tpHl$, where we define $0\times\infty = \infty\times 0 = 0$, as is countable supremum $\sup:\tpHl^\omega \fto \tpHl$, $(x_0,x_1,x_2,\ldots)\mapsto\sup\{x_0,x_1,x_2,\ldots\}$.
Further, $\abs:\tpRe\fto\tpHl$ is computable, as is the embedding $\tpSi\hookrightarrow\tpHl$ taking $\true \mapsto 1$ and $\indet \mapsto 0$.
We let $\I_<$ be the unit interval $[0,1]$, again with the topology of lower convergence with open sets $(a,1]$ for $a\in[0,1)$ and $\I$ itself, and $\I_>$ the interval with the topology of upper convergence.

A \emph{computable metric space} is a pair $(\tpX,d)$ where $\tpX$ is a computable type, and $d:\tpX\times \tpX\fto \R^+$ is a computable metric, such that the extension of $d$ to $\tpX\times\clset(\tpX)$ (where $\clset(\tpX)$ is the type of closet subsets of $\tpX$) defined by $d(x,A)=\inf\{d(x,y) \mid y\in A\}$ is computable as a function into $\R^{+,\infty}_<$.
This implies that given an open set $U$ we can compute $\epsilon>0$ such that $B_{\epsilon}(x) \subset U$, which captures the relationship between the metric and the open sets.
The effective metric spaces of~\cite{Weihrauch1999} are a concrete class of computable metric space.

A type $\tpX$ is \emph{effectively separable} if there is a computable function $\xi:\N\fto\tpX$ such that $\rng(\xi)$ is dense in $\tpX$.

Throughout this paper we shall use the term ``compute'' to indicate that a formula or procedure can be effectively carried out in the framework of type-two effectivity.
Other definitions and equations may not be possible to verify constructively, but hold from axiomatic considerations.

\section{Valuations}
\label{sec:valuations}

The main difficulty with classical measure theory is that Borel sets and Borel measures have very poor computability properties.
Although a computable theory of Borel sets was given in~\cite{Brattka2005}, the measure of a Borel set is in general not computable in $\R$.
However, we can consider an approach to measure theory in which we may only compute the measure of \emph{open} sets.
Since open sets are precisely those which can be approximated from inside, we expect to be able to compute lower bounds for the measure of an open set, but not upper bounds.
The above considerations suggest an approach which has become standard in computable measure theory, namely that using \emph{valuations}~\cite{JonesPlotkin1989,Edalat1995DTI,SchroederSimpson2006ENTCS,Escardo2009}.
\begin{definition}[Valuation]
\label{defn:valuation}
The type of \emph{valuations} on $\tpX$ is the subtype $\tpOp(\tpX) \fto \tpHl$ consisting of elements $\nu$ satisfying $\nu(\emptyset)=0$ and the \emph{modularity} condition $\nu(U)+\nu(V) = \nu(U\cup V) + \nu(U\cap V)$ for all $U,V \in \opset(\tpX)$.
\end{definition}
Note that since our valuations are elements of $\tpOp(\tpX)\fto\tpHl$, any $\nu$ satisfies the \emph{monotonicity} condition $\nu(U)\leq \nu(V)$ whenever $U\subset V$, and the \emph{continuity} condition $\nu\bigl(\bigcup_{n=0}^{\infty}U_n\bigr) = \lim_{n\to\infty} \nu(U_n)$ whenever $U_n$ is an increasing sequence of open sets.

A valuation $\nu$ on $\tpX$ is \emph{finite} if $\nu(\tpX)$ is finite, \emph{effectively finite} if $\nu(\tpX)$ is a computable real number, and \emph{locally finite} if $\nu(U)<\infty$ for any $U$ which is a subset of a compact set.

An effectively finite valuation computably induces an upper-valuation on closed sets $\bar{\nu}:\clset(\tpX)\fto \tpRe^{+}_{>}$ by $\bar{\nu}(A) = \nu(\tpX) - \nu(\tpX\setminus A)$.
For any finite valuation, $\nu(U)\leq \bar{\nu}(A)$ whenever $U\subset A$. We say a set $S$ is \emph{$\nu$-regular} if $\bar{\nu}(\partial S)=0$.
An open set $U$ is $\nu$-regular if, and only if, $\nu(U)=\bar{\nu}(\overline{U})$.

The following result shows that the measure of a sequence of small sets approaches zero.
Recall that a space $\tpX$ is \emph{regular} if for any point $x$ and open set $U$, there exists an open set $V$ and a closed set $A$ such that $x\in V\subset A\subset  U$.
\begin{lemma}\label{lem:decreasingopen}
Let $\tpX$ be a separable regular space, and $\nu$ a finite valuation on $\tpX$.
If $U_n$ is any sequence of open sets such that $U_{n+1}\subset U_n$ and $\bigcap_{n=0}^{\infty} U_n=\emptyset$, then $\nu(U_n)\to0$ as $n\to\infty$.
\end{lemma}


A link with classical measure theorey is provided by a number of results that show that valuations can be extended to measures on the Borel $\sigma$-algebra.
\begin{theorem}\label{thm:valuationborelmeasure}
Borel measures and continuous valuations are in one-to-one correspondance:
\begin{enumerate}
 \item on a countably-based locally-compact Hausdorff space~\cite[Corollary~5.3]{Edalat1995DSM}, or
 \item on a locally compact sober space~\cite{AlvarezManilla2002}.
\end{enumerate}
\end{theorem}
\noindent
For a purely constructive approach valuations themselves are main objects of study, and we only (directly) consider the measure of open and closed sets.

Just as for classical measure theory, we say (open) sets $U_1,U_2$ are \emph{independent} if $\nu(U_1\cap U_2)=\nu(U_1)\nu(U_2)$.
\begin{definition}[Conditioning]
Given a sub-topology $\calV$ on $\tpX$ and a valuation $\nu$ on $\tpX$, a \emph{conditional valuation} is a function $\nu(\cdot|\cdot):\tpOp(\tpX)\times\calV\fto\tpHl$ such that $\nu(U\cap V)=\nu(U|V)\nu(V)$ for all $U\in\tpOp(\tpX)$ and $V\in\calV$.
\end{definition}
\noindent
Clearly, $\nu(U\cap V)$ can be computed given $\nu(U|V)$ and $\nu(V)$.
The conditional valuation $\nu(\cdot|V)$ is uniquely defined if $\nu(V)\neq0$.
However, since $\nu(U\cap V):\R^+_<$ but $1/\nu(V):\R^{+,\infty}_>$, the conditional valuation $\nu(\cdot|V)$ cannot be \emph{computed}, even when $\nu(V)>0$, unless we are also given a set $A\in\tpCl(\tpX)$ such that $V\subset A$ and $\bar{\nu}(A\setminus V)=0$, in which case we have $\nu(U|V)=\nu(U\cap V)/\bar{\nu}(A)$.

We can define a notion of integration for positive lower-semicontinuous functions by the \emph{Choquet} or \emph{horizontal} integral; see~\cite{Tix1995,Lawson2004,Vickers2008}.
\begin{definition}[Lower integral]
\label{defn:lowerintegral}
Given a valuation $\nu:(\tpX\fto\tpSi)\fto\tpHl$, define the lower integral $(\tpX\fto\tpHl)\fto\tpHl$ by
\begin{multline} \textstyle \int_{\tpX}\psi\,d\nu = \sup \bigl\{ {\textstyle\sum_{m=1}^{n}} (p_{m}-p_{m-1}) \, \nu(\psi^{-1}(p_{m},\infty]) \\ \mid (p_0,\ldots,p_n) \in \Q^* \wedge 0 = p_0 < p_1 < \cdots < p_n \bigr\}  \label{eq:lowerintegral} \end{multline}
which is equivalent to the real integral
\begin{equation} \textstyle \int_{\tpX}\!\psi\,d\nu = \int_0^\infty \nu\bigl(\psi^{-1}(x,\infty]\bigr) dx . \label{eq:lowerrealintegral} \end{equation}
\end{definition}
\noindent

\noindent
Note that we could use any dense set of computable positive real numbers, such as the dyadic rationals $\Q_2$, instead of the rationals in~\eqref{eq:lowerintegral}.
Since each sum is computable, and the supremum of countably many elements of $\tpHl$ is computable, the lower integral is computable.

It is fairly straightforward to show that the integral is linear,
\begin{equation} \textstyle \int_\tpX (a_1\psi_1+a_2\psi_2)\,d\nu= a_1\int_\tpX\psi_1\,d\nu+a_2\int_\tpX\psi_2\,d\nu  \end{equation}
for all $a_1,a_2\in \Hl$ and $\psi_1,\psi_2:\tpX\fto\Hl$.

If $\chif_U$ is the characteristic function of a set $U$, then
\( \textstyle \int_\tpX \chif_U \,d\nu= \nu(U) ,  \)
and it follows that if $\phi=\sum_{i=1}^{n} a_i\,\chif_{U_i}$ is a step function, then
\( \textstyle \int_\tpX\phi\,d\nu = \sum_{i=1}^{n} a_i\,\nu(U_i) .\)

Given a (lower-semi)continuous linear functional $\mu:(\tpX\fto\tpHl)\fto\tpHl$, we can define a function $\tpOp(\tpX)\fto\tpHl$ by
\( U \mapsto \mu(\chif_U) \) for $U\in\tpOp(\tpX)$.
By linearity, \[ \textstyle \mu(\chif_U)+\mu(\chif_V) = \mu(\chif_{U\cap V})+\mu(\chif_{U\cup V}) .  \]
Hence $\mu$ induces a valuation on $\tpX$.
We therefore obtain a computable equivalence between the type of valuations and the type of positive linear lower-semicontinuous functionals:
\begin{theorem}
\label{thm:lowermeasurevaluation}
The type of valuations $(\tpX\fto\tpSi)\fto\tpHl$ is computably equivalent to the type of continuous linear functionals $(\tpX\fto\tpHl)\fto\tpHl$.
\end{theorem}
Types of the form $(\tpX\fto \tpT)\fto \tpT$ for a fixed type $\tpT$ form a \emph{monad}~\cite{Street1972} over $\tpX$, and are particularly easy to work with.

In~\cite[Section~4]{Edalat1995DTI}, a notion of integral $\ctsfn_\mathrm{bd}(\tpX;\R)\fto\R$ on continuous bounded functions was introduced based on the approximation by measures supported on finite sets of points.
Our lower integral on positive lower-semicontinuous functions can be extended to bounded functions as follows:
\begin{definition}[Bounded integration]
\label{defn:boundedintegral}
A continuous function $f:\tpX\fto\tpRe$ is \emph{effectively bounded} if there are (known) computable reals $a,b\in\tpRe$ such that $a < f(x) < b$ for all $x\in\tpX$.

If $\nu$ is effectively finite with $\nu(\tpX)=c$, we define the integral $\tpCts_{\mathrm{bd}}(\tpX;\tpRe) \fto \tpRe$ by
\[ \textstyle \int_\tpX f(x) \, d\nu(x) = \int_\tpX \bigl(a+f(x)\bigr) \,d\nu(x)-a\,c = b\,c - \int_\tpX \bigl(b-f(x)\bigr) \,d\nu(x)  \]
where $a<b$ are bounds for $f$.
\end{definition}
It is clear that the first formula for the integral of $f$ is computable in $\tpRe_<$ and the second in $\tpRe_>$, and that the lower and upper integrals agree if $f$ is continuous.
If $\tpX$ is compact, then any (semi)continuous function is effectively bounded, so the integrals always exist.

In order to define a valuation given a positive linear functional $\tpCts_{\mathrm{cpt}}(\tpX;\tpRe)\fto\tpRe$ on compactly-supported continuous functions, we need some way of approximating the characteristic function of an open set by continuous functions.
If $\tpX$ is \emph{effectively regular}, then given any open set $U$, we can construct an increasing sequence of closed sets $A_n$ such that $\bigcup_{n\to\infty} A_n=U$.
Further, a type $\tpX$ is \emph{effectively quasi-normal} if given disjoint closed sets $A_0$ and $A_1$, we can construct a continuous function $\phi:\tpX\fto[0,1]$ such that $\phi(A_0)=\{0\}$ and $\phi(A_1)=\{1\}$ using an effective Uryshon lemma; see~\cite{Schroeder2009} for details.

We then have an effective version of the Riesz representation theorem:
\begin{theorem}\label{thm:effectiveriesz}
Suppose $\tpX$ is an effectively regular and effectively quasi-normal type.
Then type of locally-finite valuations $(\tpX\fto\tpSi)\fto\tpHl$ is effectively equivalent to the type of positive linear functionals $\tpCts_{\mathrm{cpt}}(\tpX\fto\R)\fto\tpRe$ on continuous functions of compact support.
\end{theorem}

We consider lower-semicontinuous functionals $(\tpX \fto \tpHl) \fto \tpHl$ to be more appropriate as a foundation for computable measure theory than the continuous functionals $(\tpX \fto \tpRe) \fto \tpRe$, since the equivalence given by Theorem~\ref{thm:lowermeasurevaluation} is entirely independent of any assumptions on the type $\tpX$ whereas the equivalence of Theorem~\ref{thm:effectiveriesz} requires extra properties of $\tpX$ and places restrictions on the function space.

\begin{theorem}[Fubini]
If $\tpX_1$ and $\tpX_2$ are countably-based spaces, then for any $\psi:\tpX_1\times \tpX_2\fto \tpHl$,
\begin{equation} \textstyle \int_{\tpX_1} \int_{\tpX_2} \psi(x_1,x_2) d\nu_2(x_2) d\nu_1(x_1) = \int_{\tpX_2} \int_{\tpX_1} \psi(x_1,x_2) d\nu_1(x_1) d\nu_2(x_2) . \end{equation}
\end{theorem}
\noindent
Extending valuations to functions $(\tpX\fto\tpHl)\fto\tpHl$, we can write
\[ \nu_1(\lambda x_1.\nu_2(\lambda x_2.\psi(x_1,x_2)))=\nu_2(\lambda x_2.\nu_1(\lambda x_1.\psi(x_1,x_2))) . \]
\begin{definition}[Product valuation]
Let $\nu_i$ be a valuation on $\tpX_i$ for $i=1,2$, where each $\tpX_i$ is countably-based.
The \emph{product} of two valuations is given by
\begin{equation}\begin{aligned} \textstyle [\nu_1 \times \nu_2](U) & \textstyle := \int_{\tpX_1} \int_{\tpX_2} \chi_U(x_1,x_2) d\nu_2(x_2) d\nu_1(x_1) \\ &\textstyle \qquad = \int_{\tpX_2} \int_{\tpX_1} \chi_U(x_1,x_2) d\nu_1(x_1) d\nu_2(x_2) , \end{aligned}\end{equation}
where the two integrals are equal by Fubini's theorem.
\end{definition}


In the sequel, we shall make frequent use of the following result.
\begin{proposition}
\label{prop:subsetintersectionunion}
Let $U,V,W:\tpOp(\tpX)$ with $U\subset V$, and $\nu$ a valuation on $\tpX$. Then
\begin{enumerate}[(a)]
\item\label{itm:subsetintersection} $\nu(U)+\nu(V\cap W) \leq \nu(V)+\nu(U\cap W)$, and
\item\label{itm:subsetunion} $\nu(U)+\nu(V\cup W) \leq \nu(V) + \nu(U\cup W)$.
\end{enumerate}
\end{proposition}
\noindent
The proof is straightforward.

\section{Lower-measurable sets}
\label{sec:measurablesets}

In this section, we define measures of non-open sets.

The standard approach to probability theory used in classical analysis is to define a measure over a $\sigma$-algebra of sets.
The main difficulty with a direct effectivisation of probability theory via $\sigma$-algebras is the operation on complementation.
Given an open set $U$, we can \emph{only} hope to compute $\nu(U)$ in $\tpRe^+_<$ (i.e. from below), and since $\sigma$-algebras are closed under complementation, for $A=U^\complement$, we compute $\nu(A)$ in $\tpRe^+_>$.
Then for a countable union of nested closed sets $A_\infty=\bigcup_{n=0}^{\infty} A_n$ with $A_{n+1}\subset A_n$, we need to find information about the limit $\lim_{n\to\infty}\nu(A_n)$ which is an increasing sequence in $\tpRe^+_>$, so we can find neither an upper- or a lower-bound.

Our solution is to consider, for a \emph{fixed} probability valuation $\nu$, a type of \emph{$\nu$-lower-measurable sets}.
These essentially extend the open sets to the $G_\delta$ sets, and have a representation under which it is possible to compute the $\nu$-measure from below.
Further, they are closed under finite intersection and countable union (so may be called a \emph{$\sigma$-semiring}, though this usage is different from that of e.g.~\cite{Schmets2004PP}).
However, they do not formally define a topology on $\tpX$, since they are essentially \emph{equivalence-classes} of subsets of $\tpX$ (and even this property only holds for sufficiently nice spaces, including the Polish spaces).
The resulting theory can be seen as a construction of an outer-regular measure for $G_\delta$ subsets of a space (see~\cite{vanGaansPP}.

\subsection{Definition and basic properties}

We first define the type of lower-measurable sets, and prove some of its basic properties.

\begin{definition}[Lower-Cauchy sequence]
\label{defn:lowermeasureablesequence}
Let $\calV$ be a set with a intersection (or \emph{meet}) operation $\cap:\calV\times\calV\to\calV$ and a compatible subset (or order) relation $\subset$ on $\calV\times\calV$, and let $\nu:\calV\to\R$.
\par
A sequence $(V_k)$ of elements of $\calV$ is a \emph{lower-Cauchy sequence} if for all $\epsilon>0$, there exists $n=N(\epsilon)$, such that
\begin{equation}\label{eq:effectivelowermeasurable} \forall m>n, \ \nu(V_m\cap V_n)\geq \nu(V_n)-\epsilon .  \end{equation}
\par
The convergence is \emph{effective} if $N(\epsilon)$ is known, equivalently, if there is a known sequence $(\epsilon_k)$ with $\lim_{k\to\infty}\epsilon_k=0$ such that for all $m>n$, $\nu(V_m\cap V_n)\geq \nu(V_n)-\epsilon_n$. The convergence is \emph{fast} if $\epsilon_k=2^{-k}$.
\par
The sequence is \emph{monotone (decreasing)} if for all $m>n$, $V_m\subset V_n$.
\end{definition}
\noindent
If $(V_n)$ is an effective lower-Cauchy sequence, then
\[ \forall m>n, \ \nu(V_n\setminus V_m) \leq \epsilon_n . \]
If $(V_n)$ is a fast monotone lower-Cauchy sequence, then
\[ \forall m>n, \ \nu(V_n) \geq \nu(V_m)\geq \nu(V_n)- 2^{-n} . \]

\begin{definition}[Equivalence of lower-Cauchy sequence]
\label{defn:equivalentlowermeasureablesequences}
Two $\nu$-lower-Cauchy sequences are \emph{equivalent}, denoted $(U_n)\sim (V_n)$ if, and only if, for
\[ \forall \epsilon>0,\;\exists n\in\N,\;\forall m \geq n,\ \nu(U_m\cap V_m) \geq \max(\nu(U_m),\nu(V_m)) - \epsilon . \]
\end{definition}

\begin{lemma}
If $(U_n)$ and $V_n$ are fast lower-Cauchy sequences, then $(U_n)$ and $(V_n)$ are equivalent if, and only if
\[ \forall n\in\N, \ \nu(U_n\cap V_n) \geq \max(\nu(U_n),\nu(V_n)) - 2^{-n} . \]
\end{lemma}

\begin{lemma}
The relation $\sim$ is an equivalence relation on fast lower-Cauchy sequences.
\end{lemma}
\begin{proof}
Reflexivity and commutativity are immediate; it remains to show transitivity.
Suppose $(U_n)$, $(V_n)$ and $(W_n)$ are fast monotone $\nu$-lower-Cauchy sequences, that $(U_n)\sim(V_n)$ and $(V_n)\sim (W_n)$. Then for any $m>n$, $\nu(U_n\cap W_n)\geq\nu(U_m\cap V_m\cap W_m)\geq \nu(U_m\cap V_m)+\nu(V_m\cap W_m)-\nu(W_m) \geq \nu(U_m)+\bigl(\nu(U_m\cap V_m)-\nu(U_m)\bigr)+\bigl(\nu(V_m\cap W_m)-\nu(V_m)\bigr)\geq \nu(U_m)-2\!\times\!2^{-m} \geq \nu(U_n)-2^{-n}-2\!\times\!2^{-m}$. Since $m$ can be made arbitrarily large, $\nu(U_n\cap W_n)\geq \nu(U_n)-2^{-n}$ as required. By symmetry, $\nu(U_n\cap W_n)\geq \nu(W_n)-2^{-n}$.
\end{proof}

\begin{definition}[Lower-measurable set]
\label{defn:lowermeasurableset}
Let $\tpX$ be a type, and $\nu$ a valuation on $\tpX$.
The type of $\nu$-lower-measurable sets is defined as the equivalence classes of fast monotone $\nu$-lower-Cauchy sequences of open subsets of~$\tpX$.
\par
Note that a fast monotone $\nu$-lower-Cauchy sequence $(U_n)$ satisfies
\begin{equation}
\forall n\in\N, \; \forall m>n,\ U_m\subset U_n \wedge \nu(U_m)\geq \nu(U_n)-2^{-n} .
\end{equation}
\end{definition}

By countable-additivity, the classical measure of $U_\infty$ coincides with its $\nu$-lower-measure.
\begin{proposition}[Lower-measure is computable]
If $(U_n)$ is a fast monotone lower-Cauchy sequence of open sets, then $\sup_{n\in\N}(\nu(U_n)-2^{-n})$ is computable in $\R^+_<$.
\end{proposition}
\begin{proof}
The value $\sup_{n\in\N}\nu(U_n)-2^{-n}$ is the supremum of a countable set of lower-reals, so is computable in $\tpR^+_<$.
\end{proof}
\begin{proposition}[Lower-measure is well-defined]
If $(U_n)$ and $(V_n)$ are equivalent fast monotone $\nu$-lower-Cauchy sequences, then $\sup_{n\in\N}\bigl(\nu(U_n)-2^{-n}\bigr)=\sup_{n\in\N}\bigl(\nu(V_n)-2^{-n}\bigr)$.
\end{proposition}
\begin{proof}
Suppose $U_n$ and $V_n$ are equivalent $\nu$-lower-Cauchy sequences.
Then for all $n$, we have $\nu(U_\infty)\geq \nu(U_n)-2^{-n} \geq \nu(U_n\cap V_n)-2^{-n} \geq \nu(V_n)-2\times 2^{-n} \geq \nu(V_\infty)-2\times 2^{n}$. Since $n$ is arbitrary, $\nu(U_\infty) \geq \nu(V_\infty)$. Switching the $U$s and $V$s gives  $\nu(U_\infty)\geq \nu(V_\infty)$.
\end{proof}

\begin{definition}[Lower-measure]
\label{defn:lowermeasure}
The $\nu$-lower-measure of a lower measurable set $U_\infty$ is defined as $\sup_{n\in\N}\nu(U_n)-2^{-n}$, where $U_n$ is any fast monotone lower-Cauchy sequence converging to $U_\infty$.
\end{definition}

\begin{lemma}
\label{lem:subsequence}
Let $(U_n)$ be a $\nu$-lower-Cauchy sequence.
Then sequence $(U_{n+1})$ is $\nu$-lower-Cauchy, and $(U_{n+1})\sim_\nu (U_n)$.
\end{lemma}
\begin{proof}
If $m>n$, $\nu(U_{n+1})-\nu(U_{m+1}) \leq 2^{-(n+1)} < 2^{-n}$.
\end{proof}

We first show that given an effective lower-Cauchy sequence, then it is possible to compute an equivalent fast lower-Cauchy subsequence, and given a non-monotone lower-Cauchy sequence of lower-measurable sets (notably, of open sets), we can compute a monotone sequence with the same limit.
\begin{lemma}[Computing fast monotone lower-Cauchy sequences]
\label{lem:fasteffectivelowermeasurablesequence}
\mbox{}
\begin{enumerate}
\item Suppose $(V_n)_{n\in\N}$ is an effective $\nu$-lower-Cauchy sequence.
Then  $U_n=V_{N(2^{-n})}$ is an equivalent fast lower-measurable subsequence.
\item Suppose $(V_n)_{n\in\N}$ is a fast $\nu$-lower-Cauchy sequence.
Then $U_n = \bigcup_{m\geq n+1} V_m$ is an equivalent fast monotone $\nu$-lower-Cauchy sequence.
\end{enumerate}
\end{lemma}
\begin{proof}\mbox{}
\begin{enumerate}
\item Clearly for $m>n\geq N(2^{-n})$, we have $\nu(U_m\cap U_n) = \nu(V_{N(2^{-m})}\cap V_{N(2^{-n})}) \leq \nu(V_{N(2^{-m})}) + 2^{-n} = \nu(U_m)+2^{-n}$.
\item Since $U_n=U_{n+1}\cup V_{n+1}$, we have $\nu(U_n) + \nu(U_{n+1}\cap V_{n+1})=\nu(U_{n+1}\cup V_{n+1}) + \nu(U_{n+1}\cap V_{n+1})= \nu(U_{n+1})+\nu(V_{n+1})$, and since $V_{n+2}\subset U_{n+1}$, we have $\nu(U_{n+1}\cap V_{n+1})\geq \nu(V_{n+2}\cap V_{n+1})\geq \nu(V_{n+1})-2^{-(n+1)}$.
Hence $\nu(U_n) \leq \nu(U_{n+1}) + 2^{-(n+1)}$, and by induction, we see $\nu(U_n) < \nu(U_{m}) + 2^{-n}$ wheneve $m>n$.
Hence $(U_n)$ is a fast monotone lower-Cauchy subsequence.

To show $(U_n)\sim(V_n)$, since clearly $\nu(U_n)\geq\nu(V_n)$, we need to show $\nu(U_n)\leq \nu(V_n)+\epsilon$ for $n$ sufficiently large.
We first show that for all $m>n$, $\nu(U_n)\leq 2^{-n} + \nu(V_m)$.
Note $U_n=U_{n+1}\cup V_{n+1}$, and $V_{n+2}\subset U_{n+1}$.
Let $U_{n,m}=\bigcup_{k=n+1}^{m}V_k$, noting $U_{m-1,m}=V_m$.
Then $\nu(U_{n+1,m})+\nu(V_{n+1})=\nu(U_{n+1,m}\cup V_{n+1})+\nu(U_{n+1,m}\cap V_{n+1})=\nu(U_{n,m})+\nu(U_{n+1,m}\cap V_{n+1})\geq \nu(U_{n,m})+\nu(V_{n+2}\cap V_{n+1}) \geq \nu(U_{n,m})+\nu(V_{n+1})-2^{-(n+1)}$, so $\nu(U_{n+1,m})\geq \nu(U_{n,m})-2^{-(n+1)}$.
Hence $\nu(U_{n,m})\leq \nu(U_{m-1,m})+\sum_{k=n+1}^{m-1}2^{-k} \leq \nu(V_m)+2^{-n}$.
Since $U_n\bigcup_{m=n+1}^{\infty} U_{n,m}$, there exists $m$ such that $\nu(U_{n,m})\geq \nu(U_{n})-2^{-n}$.
Then $\nu(V_m)\geq \nu(U_n)-2\!\times\!2^{-n}$, and since $\nu(U_m)\geq \nu(U_n)-2^{-n}$, we have $\nu(V_m)\geq \nu(U_m)-3\!\times\!2^{-n}$.
\end{enumerate}
\end{proof}

\begin{remark}[Equivalent definitions of $\nu$-lower-measurable sets]
By Lemma~\ref{lem:fasteffectivelowermeasurablesequence}, we see that the monotonicity condition on the open sets $U_n$ in Definition~\ref{defn:lowermeasurableset} is unnecessary, and that fast convergence can be weakened to effective convergence.

By Theorem~\ref{thm:fastcauchylowermeasurableset}, we see that a second definition would be to say (i) any open set is $\nu$-lower-measurable, and (ii) any fast $\nu$-lower-Cauchy sequence $(V_n)_{n\in\N}$ of $\nu$-lower-measurable sets defines a lower-measurable set $V_\infty$, with
 then $V_\infty:=\sup_{n\in\N}\nu(V_n)-2^{-n}$.

A third definition for countably-based topological spaces would be to take a countable basis, and consider $\nu$-lower-Cauchy sequences of \emph{finite} unions of basic sets.
\end{remark}

Which definition to take is a matter of taste; our Definition~\ref{defn:lowermeasurableset} provides strong properties of the approximating sequences, so is easy to use in hypotheses, but it requires more work to prove.
However, Lemma~\ref{lem:fasteffectivelowermeasurablesequence} shows that it suffices to compute an effectively lower-Cauchy sequence.
One advantage of using fast sequences over effective sequences is that we do not have to explicitly pass around a counvergence rate.

has the advantage of providing a uniform construction, though for explicit computations, it may be more appropriate to restrict to finite unions of basic open sets.

\subsection{Intersections and unions of lower-measurable sets}

We now consider computability of intersections and unions of lower-measurable sets.

The following lemma compares measures of unions and intersections of equivalent sets.
\begin{lemma}\mbox{}
Let $(U_n)$, $(V_n)$ and $(W_n)$ be $\nu$-lower-Cauchy sequences.
\begin{enumerate}
\item If $\nu\bigl((U\cap V)_\infty\bigr)=\nu(U_\infty)=\nu(V_\infty)$, or $\nu\bigl((U\cup V)_\infty\bigr)=\nu(U_\infty)=\nu(V_\infty)$, then $(U_n)\sim(V_n)$.
\item If $(U_n)\sim_\nu (V_n)$, then $(U_n\cap V_n)\sim_\nu (U_{n+1}\cup V_{n+1})$, and both are equivalent to $(U_n)$ and $(V_n)$.
\end{enumerate}
\end{lemma}
\begin{proof}\mbox{}
\begin{enumerate}
\item
$\nu(U_n\cap V_n)\geq \lim_{n\to\infty}\nu(U_n\cap V_n) = \nu(U_\infty) \geq \nu(U_n)-2^{-n}$.
Similarly, $\nu(U_n\cup V_n)\geq \nu(\lim_{n\to\infty}U_n\cup V_n)  = \nu(U_\infty) \geq \nu(U_n)-2^{-n}$.
\item
$\nu(U_n\cap(U_n\cap V_n))=\nu(U_n\cap V_n)$, and $\nu(U_n\cap V_n) \geq \nu(U_n)-2^{-n}$ by definition of $\sim_\nu$.
$\nu(U_n\cap(U_n\cup V_n))=\nu(U_n)$, and $\nu(U_n\cup V_n)=\nu(U_n)+\nu(V_n)-\nu(U_n\cap V_n)\leq \nu(U_n)+2^{-n}$, so $\nu(\nu(U_n\cap(U_n\cup V_n))\geq \nu(U_n\cup V_n)-2^{-n}$ as required.
\qedhere
\end{enumerate}
\end{proof}

\begin{theorem}[Intersections and unions of lower-measurable sets]
\label{thm:intersectionunion}
Let $\nu$ be a probability valuation on a type $\tpX$.
Then operations of (1) intersection, and (2) union are computable on $\nu$-lower-measurable sets.
\end{theorem}
\noindent
\begin{proof}
Suppose $U_n \lmsto_\nu U_\infty$ and $V_n \lmsto_\nu V_\infty$.
\begin{enumerate}
\item We show that $(U_n \cap V_n)$ is a fast $\nu$-lower-Cauchy sequence.
For $m>n$ and any $l>m$, $\nu(U_n\cap V_n)-\nu(U_m\cap V_m) \leq \nu(U_n)-\nu(U_l\cap V_l) \leq \nu(U_n)-\nu(U_l)+2^{-l} \leq 2^{-n}+2^{-l}$. Since $l$ can be made arbitrarily large, $\nu(U_m\cap V_m)\geq \nu(U_n\cap V_n)-2^{-n}$.
\item
For $m>n$, we have $U_m\subset U_n$ and $V_m\subset V_n$, so by Proposition~\ref{prop:subsetintersectionunion}, $\nu(U_m)+\nu(V_m)+\nu(U_n\cup V_n)  \leq \nu(V_m) + \nu(U_m\cup V_n) + \nu(U_n)  \leq \nu(U_m \cup V_m) + \nu(U_n)+\nu(V_n)$.
Hence $\nu(U_n\cup V_n)-\nu(U_m\cup V_m) \leq (\nu(U_n)-\nu(U_m))+(\nu(V_n)-\nu(V_m)) \leq 2\!\times\!2^{-n}$, so $\nu(U_{n+1}\cup V_{n+1})-\nu(U_{m+1}\cup V_{m+1})\leq 2^{-n}$.
\qedhere
\end{enumerate}
\end{proof}
\noindent

\begin{theorem}[Countable unions of lower-measurable sets]
\label{thm:countableunion}
Let $\nu$ be a probability valuation on a type $\tpX$.
Then the operation of countable union is computable on $\nu$-lower-measurable sets.
\end{theorem}
\begin{proof}
We can show that if $U_k\subset V_k$ for $k=1,2,\ldots$, then $\nu(\bigcup_{k=1}^{\infty}V_k)-\nu(\bigcup_{k=1}^{\infty}U_k)\leq \sum_{k=1}^{n}\bigl(\nu(V_k)-\nu(U_k)\bigr)$.
\par
Given $U_{k,n} \lmsto_\nu U_k$, we show that $V_n:=\bigcup_{k=0}^{\infty} U_{k,n+k+1}$ is a $\nu$-lower-Cauchy sequence.
For if $m>n$, then $\nu\bigl(\bigcup_{k=0}^{\infty} U_{k,n+k+1}\bigr)-\nu\bigl(\bigcup_{k=0}^{\infty}U_{k,m+k+1}\bigr) \leq \sum_{k=0}^{\infty}\bigl(\nu(U_{k,n+k+1})-\nu(U_{k,m+k+1})\bigr)\leq\sum_{k=0}^{\infty}2^{-(n+k+1)}=2^{-n}$ as required.
\end{proof}

We can similarly compute countable intersections of \emph{effectively} decreasing sequences of $\nu$-lower-measurable sets.
\begin{theorem}[Effective countable intersections of lower-measurable sets]
\label{thm:fastcauchylowermeasurableset}
Suppose $(V_n)$ is a fast monotone lower-Cauchy sequence of $\nu$-lower-measurable sets.
Then $V_n$ converges effectively to a $\nu$-lower-measurable set $V_\infty$.
\end{theorem}
\begin{proof}
Write $V_{k,\infty}=V_k$ and let $V_{k,n}$ be a fast monotone lower-Cauchy sequence of open sets converging to $V_k$.
Define $V_{\infty,n}=V_{n+1,n+1}\cap V_{\infty,n-1}$, which is clearly monotone.
Note that since $V_{n+1,\infty}\subset V_{n+1,n+1}$ and $V_{n+1,\infty}\subset V_{n,\infty}$, we have $V_{\infty,n}\supset V_{n+1,\infty}=V_{n+1}$.
Then for $m>n$, $\nu(V_{\infty,n})-\nu(V_{\infty,m})\leq \nu(V_{n+1,n+1})-\nu(V_{m+1})\leq \nu(V_{n+1,n+1})-\nu(V_{n+1,\infty})+\nu(V_{n+1})-\nu(V_{m+1})\leq 2^{-(n+1)}+2^{-(n+1)}=2^{-n}$, so $V_{\infty,n}$ is a fast monotone lower-Cauchy sequence of open sets, so represents a lower-measurable set $V_{\infty,\infty}=V_\infty$.
\par
Finally, we have $\nu(V_n)-\nu(V_\infty)-\nu(V_{n,\infty})-\nu(V_{\infty,\infty})\leq \nu(V_{\infty,n})-\nu(V_{\infty,\infty})\leq 2^{-n}$, from which we see that $V_\infty$ is indeed the limit of $(V_n)$.
\end{proof}

\begin{proposition}[Lower measure is modular]
$\nu(U_\infty\cap V_\infty)+\nu(U_\infty\cup V_\infty) = \nu(U_\infty)+\nu(V_\infty)$.
\end{proposition}
\begin{proof}
For fixed $n$, $\nu(U_n\cap V_n)+\nu(U_n\cup V_n) = \nu(U_n)+\nu(V_n)$ by modularity of valuations.
Then $\nu(U_\infty\cap V_\infty)+\nu(U_\infty\cup V_\infty) \geq \nu(U_n\cap V_n)-2^{-n}+\nu(U_n\cup V_n)-2\!\times\!2^{-n}=\nu(U_n)+\nu(V_n)-3\times2^{-n}\geq \nu(U_\infty)+\nu(V_\infty)-3\times2^{-n}$. Taking $n\to\infty$ gives $\nu(U_\infty\cap V_\infty)+\nu(U_\infty\cup V_\infty) \geq \nu(U_\infty)+\nu(V_\infty)$. The reverse inequality is similar, since $\nu(U_\infty)+\nu(V_\infty)\geq \nu(U_n)+\nu(V_n)-2\times2^{-n}$.
\end{proof}

\subsection{Topology of lower-measurable sets}

The representation of $\nu$-lower-measurable sets induces a (non-Hausdorff) quotient topology on the space.
Recall that for open sets, $U_n\to U_\infty \iff \forall x\in U_\infty, \exists N, \forall n\geq N, x\in U_n$.

For $\nu$-lower-Cauchy sequences, convergence is given by $(U_{k,n})\to (U_{\infty,n})$ as $k\to\infty$ if for all $n$, there exists $K$ such that $\nu(U_{k,n}\cap U_{\infty,n})\geq \nu(U_{\infty,n})-2^{-n}$ whenever $k\geq K(n)$.
The convergence is effective if $K(n)$ is known; by restricting to subsequences we may take $K(n)=n$.
In this case, a $\nu$-lower-Cauchy sequence representing the limit is $U_{\infty,n}=\bigcup_{k=n+1}^{\infty}U_{k,n+k+1}$.

\begin{property}[Topology on $\nu$-lower-measurable sets]
A set of $\nu$-lower-measurable sets $\mathcal{W}$ is \emph{open} if
\begin{equation}
\forall W\in \mathcal{W},\ \exists\epsilon>0,\ \forall V,\ \nu(V\cap W)>\nu(W)-\epsilon \implies V\in \mathcal{W}.
\end{equation}
\end{property}

\begin{property}
The convergence relation on $\nu$-lower-Cauchy sequences of $\nu$-lower-measurable sets is given by $V_k \to V_\infty$ as ${k\to\infty}$ if $\liminf_{k\to\infty}\nu(V_k\cap V_\infty)\geq\nu(V_\infty)$.
\par
The convergence is \emph{effective} if $\nu(V_k\cap V_\infty)\geq\nu(V_\infty)-\epsilon(k)$ for \emph{known} $\epsilon(k)$, and \emph{fast} if $\epsilon(k)=2^{-k}$.
\end{property}

\subsection{Relationship with classical measure-theory}

\begin{remark}[Outer-regular measures]
In the literature on classical measure theory (see~\cite{vanGaansPP}), defining $\nu(U_\infty)=\bigcap_{n\in\N}\nu(U_n)=\inf_{n\in\N} \nu(U_n)$ for a decreasing sequence of open sets corresponds to a \emph{outer-regular measure}, since we approximate $U_\infty$ from outside. However, since $U_n$ converges to $U_\infty$ from above, but the open sets $U_n$ are  inherently approximated from below, we cannot compute the measure of an \emph{arbitrary} decreasing sequence of open sets.
This motivates the use of $\nu$-lower-Cauchy sequences, which converge rapidly from above.
Since we compute $\nu(U_\infty)$ from below, we use the terminology ``lower measure''.
\end{remark}

\begin{remark}[Relationship with classical Borel measures]
Our lower-measurable sets are all Borel sets, with lower measure equal to the classical measure. Further, since any outer measure on a separable metric space is a Borel measure, the measure of any set is the infemum of the measure of its $\epsilon$-neighbourhoods, we see that any measurable set is equal to a lower-measurable set up to a set of measure $0$.
Hence our lower-measurable sets capture the measure-theoretic behaviour of all Borel sets, but do so in a way in which the measure is semicomputable.
\end{remark}

\subsection{Measurable sets}

\begin{definition}[Upper-measurable sets; upper-measure]
The type of \emph{upper-measurable sets} is the set of increasing sequences of closed sets $(A_n)$ such that $\nu(A_m) \leq \nu(A_n) + 2^{-n}$ for all $n$ and all $m>n$. The \emph{upper-measure} is $\inf_{n\in\N} \nu(A_n)+2^{-n}$.
\end{definition}
Note that a representation of an upper-measurable set is the same as the complement of a lower-measurable set.

\begin{definition}[Measurable sets]
\label{lem:measurableset}
The type of \emph{measurable sets} consists of equivalence classes of monotone sequences of pairs of open and closed sets $(U_n,A_n)$ such that $U_{n}\subset U_{n+1}\subset A_{n+1}\subset U_{n}$ for all $n$, and $\nu(A_n\setminus U_n) \leq 2^{-n}$, under the equivalence relation $(U_n,A_n)\sim (V_n,B_n)$ if, and only if, $\nu(A_n\setminus V_n \cup B_n\setminus U_n)\to0$ as $n\to\infty$.
\end{definition}
In other words type of $\nu$-measurable sets in $\tpX$ is the effective completion of the type of pairs $(U,A)\in\tpOp(\tpX)\times\tpCl(\tpX)$ satisfying $U\subset A$ under the (non-metric) distance $d((U,A),(V,B))=P(A\setminus V \cup B\setminus U)$.
Note that $d((U,A),(U,A))=\nu(A\setminus U)$, which need not be zero, but from the condition \linebreak $d((U_m,A_m),(U_n,A_n))<2^{-\min(m,n)}$, we have $\mu(A_n\setminus U_n)<2^{-n}$ for all $n$.

The type of measurable sets is equivalent to giving a fast monotone lower-Cauchy sequence $(U_n)$ and a fast monotone upper-Cauchy sequence $(A_n)$ for the same set i.e. such that $\nu(A_n\setminus U_m)\to0$ whenever $m,n\to\infty$.

\subsection{Lower-measurable sets as point-sets}

If $(U_n)$ is a $\nu$-lower-Cauchy sequence, we will write $\nu(U_\infty)$ for $\sup_{n\in\N}\nu(U_n)-2^{-n}$.
Similarly, we write $\nu((U\cap V)_\infty)$ for $\sup_{n\in\N}\nu(U_n\cap V_n)-2^{-n}$.
It is tempting to define the lower-measure of $(U_n)$ as a property of the \emph{intersection} $\bigcap_{n\in\N}U_n$.
Unfortunately, for general spaces, it need not be the case that $\sup_{n\in\N}\nu(U_n)-2^{-n} = \sup_{n\in\N}\nu(V_n)-2^{-n}$ even if $\bigcap_{n\in\N}U_n=\bigcap_{n\in\N}V_n$.
However, for spaces satisfying the conditions of Theorem~\ref{thm:valuationborelmeasure}, the $\nu$-lower-measure is indeed a property of the $G_\delta$ intersection.
We prove this directly for Polish (separable completly metrisable) spaces:

\begin{theorem}[Lower-measurable sets in metric spaces]
\label{thm:gdfesetmeasure}
If $\tpX$ is a separable complete metric space, $\nu$ a valuation on $\tpX$, and $\bigcap_{n\in\N}U_n=\bigcap_{n\in\N}V_n$ for $\nu$-lower-Cauchy sequences $(U_n),(V_n)$, then $\sup_{n\in\N}U_n-2^{-n}=\sup_{n\in\N}V_n-2^{-n}$.
\end{theorem}
\begin{proof}
Since $I_{\epsilon}(U)\subset \clI_{\epsilon}(U)\subset I_{\delta}(U)$ when $\delta<\epsilon$, and since $\nu$ is continuous, there exist closed sets $A_n\subset U_n$ such that $\nu(A_n)\geq \nu(U_n)-2^{-(n+1)}$.
Let $B_n=\bigcap_{m\geq n}A_m$, so $B_n$ is an increasing sequence of closed sets such that $B_n\subset \bigcap_{m\geq n}U_m = U_\infty$ for all $n$, so $\bigcup_{n\in\N}B_n\subset \bigcap_{n\in\N}U_n$.
Further, since $\nu(A_n\cap\cdots\cap A_m)\geq\nu(U_n\cap\cdots\cap U_m)+\sum_{l=n}^{m}\bigl(\nu(A_l)-\nu(U_l)\bigr)\geq \nu(U_m)-2^{-n}$ for all $m$, we have that $\nu(B_n)\geq\inf_{m\in\N}\nu(U_m)-2^{-n}=\sup_{m\in\N}U_m-2^{-m}-2^{-n}$, so $\inf_{n\in\N}\nu(B_n)+2^{-n}\geq\sup_{n\in\N}\nu(U_n)-2^{-n}$, and $\lim_{n\to\infty}\nu(B_n)\geq\lim_{n\to\infty}\nu(U_n)$.
Since also $B_n\subset\bigcap_{n\in\N}V_n$ for all $n$, we have $\lim_{n\to\infty}\nu(U_n)\leq\lim_{n\to\infty}\nu(B_n)\leq\lim_{n\to\infty}\nu(V_n)$, so $\sup_{n\in\N}U_n-2^{-n}\leq\sup_{n\in\N}V_n-2^{-n}$. The reverse inequality follows by symmetry.
\end{proof}



\section{Computable Random Variables}
\label{sec:randomvariables}

In the standard approach to probability theory developed in classical analysis, one defines random variables as measurable functions over a base probability space.
Given types $\tpX$ and $\tpY$, a representation of the Borel measurable functions $f$ from $\tpX$ to $\tpY$ was given in~\cite{Brattka2005}, but this does not allow one to compute lower bounds for the measure of $f^{-1}(V)$ for $V\in\tpOp(\tpY)$.

A computable theory of random variables should, at a minimum, enable us to perform certain basic operations, including:
\begin{enumerate}[(i)]
 \item Given a random variable $X$ and open set $U$, compute lower-approximation to $\Pr(X\in U)$. \label{property:distribution}
 \item Given random variables $X_1,X_2$, compute the random variable $X_1\times X_2$ giving the joint distribution. \label{property:product}
 \item Given a random variable $X$ and a continuous function $f$, compute the image $f(X)$. \label{property:image}
 \item Given a sequence of random variables $X_1,X_2,\ldots$ converging effectively in probability, compute a limit random variable $X_\infty = \lim_{m\to\infty} X_m$. \label{property:limit}
 \item Given a probability distribution $\nu$ on a sufficiently nice space $\tpX$, compute a random variable $X$ with distribution $\nu$. \label{property:realisation}
\end{enumerate}
Property~(\ref{property:distribution}) states that we can compute the distribution of a random variable, while property~(\ref{property:product}) implies that a random variable is \emph{more} than its distribution; it also allows us to compute its joint distribution with another random variable.
Property~(\ref{property:image}) also implies that for random variables $X_1,X_2$ on a computable metric space $(\tpX,d)$, the random variable $d(X_1,X_2)$ is computable in $\R^+$, so the probability $\Pr(d(X_1,X_2)<\epsilon)$ is computable in $\tpIv_<$, and $\Pr(d(X_1,X_2)\leq\epsilon)$ is computable in $\tpIv_>$.
Property~(\ref{property:limit}) is a completeness property and allows random variables to be approximated.
Property~(\ref{property:realisation}) shows that random variables can realise a given distribution.
These properties are similar to those used in~\cite{Kersting2008}.

Ideally, one would like a representation of bounded measurable functions $f:\tpX\fto \R$ such that for \emph{every} finite measure $\mu$ on $\tpX$, the integral $\int_\tpX f(x) \, d\mu(x)$ is computable.
But then $f(y) = \int_\tpX f(x) \, d\delta_y(x)$ would be computable, so $f$ would be continuous.
Any effective approach to measurable functions and integration must therefore take some information about the measure into account.

We will consider random variables on a fixed probability space $(\Omega,P)$.
Since any probability distribution on a Polish space is equivalent to a distribution on the standard Lesbesgue-Rokhlin probability space~\cite{Rokhlin1952}, it is reasonable to take the base space to be the Cantor space $\Sigma=\{0,1\}^\omega$ and $P$ the standard measure.

However, our treatment of random variables will require the notion \emph{lower-measures}, and a type of \emph{lower-measurable sets}, for which we can compute the measure in $\R^+_<$.

\subsection{Measurable functions and random variables}

\begin{definition}[Measurable function]
Let $\tpW,\tpX$ be types and $\nu$ a finite measure on $\tpW$.
Then the type of \emph{$\nu$-measurable functions} $f$ from $\tpW$ to $\tpX$ is defined by continuous $f^{-1}:\tpOp(\tpX)\to\tpLM(\tpW)$ satisfying $f^{-1}(\emptyset)=_\nu \emptyset$, $f^{-1}(X)=_\nu W$, $f^{-1}(U_1\cap U_2)=_\nu f^{-1}(U_1)\cap f^{-1}(U_2)$ and $f^{-1}(U_1\cup U_2)=_\nu f^{-1}(U_1)\cup f^{-1}(U_2)$.
\par
We denote the $\nu$-measurable functions $f$ from $\tpW$ to $\tpX$ by $f:\tpW\mfto_\nu\tpX$, or simply $f:\tpW\mfto\tpX$ if the measure $\nu$ is clear from the context.
\par
Measurable functions $f_1$ and $f_2$ are considered equal if $f_1^{-1}(U) =_\nu f_2^{-1}(U)$ for all $U\in\tpOp(\tpX)$.
\end{definition}
\noindent
Note that since $f^{-1}$ is continuous, then we must have $f^{-1}(V_1)\subset_\nu f^{-1}(V_2)$ if $V_1\subset V_2$, and $f^{-1}(\bigcup_{n=0}^{\infty}V_n) =_\nu \bigcup_{n=0}^{\infty}f^{-1}(V_n)$.
\begin{remark}
We do not actually define $f$ as a \emph{function} $\tpW\to\tpX$ since this would involve evaluating at points, though clearly any continuous function $\tpW\to\tpX$ is measurable.
\end{remark}
\noindent

\begin{definition}[Random variable]
\label{defn:randomvariable}
Let $\Omega$ be a separable complete metric space used as a base space, $P$ a probability measure on $\Omega$, and $\tpX$ a topological space
Then a \emph{random variable} $X$ on $\tpX$ is a $P$-measurable function $\Omega\mfto \tpX$.
We denote the type of random variables on $\tpX$ by $\tpRV(\tpX)$ or $\Omega\mfto\tpX$.
\end{definition}
\noindent
We will sometimes write $\P(X\in U)$ as a shorthand for $P(\{\omega\in\Omega \mid X(\omega)\in U\})$; for a piecewise-continuous random variable $X$ we implicitly restrict $\omega$ to $\dom(X)$.

Just as for measurable functions, although a random variable $X$ is \emph{defined} relative to the underlying space $\Omega$, we cannot in general actually \emph{compute} $X(\omega)$ in any meaningful sense for fixed $\omega\in\Omega$. The expression $X(\omega)$ only makes sense for random variables \emph{given} as (piecewise) continuous functions $\Omega\fto \tpX$ as stated below:
\begin{definition}[(Piecewise)-continuous random variable]\label{defn:continuousrandomvariable}\label{defn:piecewisecontinuousrandomvariable}
A \emph{continuous random variable} on $(\Omega,P)$ with values in $\tpX$ is a continuous function $X:\Omega\fto \tpX$.
\par
A \emph{piecewise-continuous random variable} is a continuous partial function $X:\Omega\pfto \tpX$ such that $\dom(X)\in\tpOp(\Omega)$ and $P(\dom(X))=1$.
\end{definition}
\noindent
We use the terminology ``piecewise-continuous'' since $X:\Omega\pfto \tpX$ may arise as the restriction of a piecewise-continuous function to its continuity set.

\begin{observation}
\label{obs:piecewisecontinuous}
Given a piecewise-continuous random variable $X:\Omega\pfto\tpX$, we can compute $X^{-1}:\tpOp(\tpX)\to\tpMe_<(\Omega)$, since $X^{-1}(U)$ is open for open $U$, and any open set is lower-measurable.
\end{observation}

Clearly, for piecewise-continuous random variables $X_1,X_2$, we have $X_1 = X_2$ if $P(\{\omega \in \Omega \mid X_1(\omega) \neq X_2(\omega)\})=0$; in other words, $X_1$ and $X_2$ are \emph{almost-surely equal}.

By~\cite[Theorem~2.2.4]{Weihrauch1999}, machine-computable functions $\{0,1\}^\omega\fto\{0,1\}^\omega$ are defined on a $G_\delta$-subsets of $\{0,1\}^\omega$.
Indeed, any function into a metric space is continuous on a $G_\delta$ set of points.
This makes functions defined and continuous on a full-measure $G_\delta$-subset of $\{0,1\}^\omega$ a natural class of random variables, where by a full-measure $G_\delta$-set $W$, we require $P(U)=1$ whenever $W\subset U$ for open $U$.
\begin{definition}[Almost-surely continuous random variable]\label{defn:almostsurelycontinuousrandomvariable}
An \emph{almost-surely-continuous random variable} on $(\Omega,P)$ with values in $\tpX$ is a continuous partial function $X:\Omega\pfto \tpX$ such that $\dom(X)$ is a $G_\delta$ set and $P(\dom(X))=1$.
\end{definition}
\noindent
Not all measurable random variables are almost-surely continuous:
\begin{example}
Define a strong Cauchy sequence $X_n$ of piecewise-continuous random variable taking values in $\{0,1\}$ such that $X_n=1$ on a decreasing sequence of closed sets $W_n$ of measure $(1+2^{-n})/2$ whose limit is a Cantor set.
Then $X_\infty=\lim_{n\to\infty}X_n$ is discontinuous on a set of positive measure, so is not an almost-surely-continuous random variable.
\end{example}

It will often be useful to consider random variables taking finitely many values:
\begin{definition}[Simple random variable]
A random variable $X$ on $(\Omega,P)$ with values in $\tpX$ \emph{simple} if it takes finitely many values.
\end{definition}

\subsection{Properties of random variables}

We now consider the properties (i)-(iv) that we wish our random variables to have, and show that they are satisfied.

\subsubsection{Distribution}

\begin{definition}[Distribution of a measurable random variable]
For a measurable random variable $X$ over base space $(\Omega,\P)$, define its \emph{distribution} by
\[ \Pr(X\in U) = \P(X^{-1}(U)) . \]
\end{definition}
From our definition, the probability distribution of a random variable is trivially computable:
\begin{observation}
Let $X$ be a random variable on $\tpX$. Then the distribution of $X$ is computable.
\end{observation}

\subsubsection{Products}

\begin{definition}
Suppose $\tpX_1$ and $\tpX_2$ are such that the product space $\tpX_1 \times \tpX_2$ is a sequential space.
Then for $X_1:\tpRV(\tpX_1)$ and $X_2:\tpRV(\tpX_2)$, the product $X_1\times X_2:\tpRV(\tpX_1\times\tpX_2)$ is defined by setting
$[X_1\times X_2]^{-1}(U_1\times U_2)=X_1^{-1}(U_1) \cap X_2^{-1}(U_2)$, and extending to arbitrary sets by taking unions of the product open sets.
\end{definition}
The product of two random variables is computable:
\begin{theorem}[Computability of products]
\label{thm:productrandomvariable}
Suppose $\tpX$ and $\tpY$ are such that the product space $\tpX \times \tpY$ is sequential.
Then for $X:\tpRV(\tpX)$ and $\tpY:\tpRV(\tpY)$, the product $X\times Y:\tpRV(\tpX\times\tpY)$ is computable.
\end{theorem}
\begin{proof}
Since $\tpX\times \tpY$ is a sequential space, the topology is generated by sets of the form $U\times V$ for $U\in\tpOp(\tpX)$ and $V\in\tpOp(\tpY)$.
Then $(X\times Y)^{-1}(U\times V) = X^{-1}(U) \cap Y^{-1}(V)$, which is computable in $\tpLM(\Omega,P)$ by Theorem~\ref{thm:intersectionunion}.
\end{proof}
\begin{lemma}
If $X_1$, $X_2$ are continuous random variables, then the product $X_1\times X_2$ is the functional product
\begin{equation*} (X_1\times X_2)(\omega) = (X_1(\omega),X_2(\omega)) .\end{equation*}
\end{lemma}
\begin{proof}
$[X_1\times X_2]^{-1}(U_1\times U_2)=X_1^{-1}(U_1) \cap X_2^{-1}(U_2)=\{\omega\in\Omega\mid X_1(\omega)\in U_1\}\times\{\omega\in\Omega\mid X_2(\omega)\in U_2\}$.
\end{proof}

\subsubsection{Image}

\begin{definition}
The image of a random variable $X:\tpRV(\tpX)$ under a continuous function $f:\tpX\to\tpY$ is defined by $[f(X)]^{-1}(V):=X^{-1}(f^{-1}(V))$.
\end{definition}

\begin{theorem}[Computability of images]
\label{thm:imagerandomvariable}
The image of a random variable under a continuous function is computable.
\end{theorem}
\begin{proof}
$[f(X)]^{-1}(V) = X^{-1}(f^{-1}(V))$, which is computable since $f^{-1}(V)$ is computable in $\tpOp(\tpX)$.
\end{proof}

\begin{remark}
If $g:\tpX\mfto\tpY$ and $f:\tpY\mfto \tpZ$ are measurable then the composition is not computable.
For example, take $g_c:\tpOmega\to\tpR$ to be the constant function $g_c(\omega)=c$, and $f:\tpR\mfto\{0,1\}$ the Heaviside function $f(x)=0$ if $x\leq 0$ and $f(x)=1$ for $x>0$.
Then for $r=1$ and $c_n\searrow c_\infty 0$ we have $P((f\circ g_{c_n})^{-1}(1))=P(\Omega)=1$, but $P((f\circ g_{c_\infty})^{-1}(1))=P(\emptyset)=0$.
\end{remark}

\subsubsection{Convergence}

The convergence relation induced by the standard representation of a function type is that of \emph{pointwise-convergence}.
For random variables (as measurable functions), this means that $X_n\to X_\infty$ if, and only if, for all open $U$, $X_n^{-1}(U)\to X_\infty^{-1}(U)$ in the type of $\P$-lower-measurable sets.
Explicitly:
\begin{property}
A sequence of random variables $(X_n)_{n\in\N}$ \emph{converges} to a random variable $X_\infty$ if for all open $U$,
$\liminf_{n\to\infty}\Pr(X_n\in U \wedge X_\infty\in U)\geq \Pr(X_\infty\in U)$.
\par
The convergence is \emph{effective} if $\Pr(X_n\in U \wedge X_\infty\in U)\geq \Pr(X_\infty\in U)-\varepsilon(U,n)$ for known $\varepsilon(U,n)$, and \emph{fast} if $\Pr(X_n\in U \wedge X_\infty\in U)\geq \Pr(X_\infty\in U)-2^{-n}$.
\end{property}
\noindent
We also obtain computability of limits of effectively-converging Cauchy-like sequences.
\begin{definition}
A sequence of random variables $(X_n)_{n\in\N}$ is a \emph{fast Cauchy sequence} if for all open $U$, and all $m>n$, $\Pr(X_m\in U \wedge X_n\in U)\geq \Pr(X_n\in U)-2^{-n}$.
\end{definition}
\begin{theorem}\label{thm:fastcauchyrandomvariable}
If $(X_n)$ is a fast Cauchy sequence of random variables, then $\lim_{n\to\infty}X_n$ exists and is computable from $(X_n)$.
\end{theorem}
\begin{proof}
For every open $U\subset \tpX$, we need to compute $X_\infty^{-1}(U)$ as a $\P$-lower-measurable set.
Define $W_n=X_n^{-1}(U)$, which is a $\P$-lower-measurable set computable from $X_n$.
Then for $m>n$, $\P(W_m\cap W_n)\geq \P(W_n)-2^{-n}$, so $(W_n)$ forms a fast lower-Cauchy sequence on $\P$-lower-measurable sets, and converges to $W_\infty$ with the correct properties by Theorem~\ref{thm:fastcauchylowermeasurableset}.
\end{proof}

\subsection{Equality of random variables}

We now show a result that two random variables are equal if, and only if, their products with the identity random variable on the base space are equal.
\begin{proposition}
\label{prop:equalproductdistribution}
Let $I$ be the identity random varible on $\Omega$. Then $X_1=X_2$ if, and only if $Y_1:=I\times X_1$ and $Y_2:=I\times X_2$ have the same distribution.
\end{proposition}
\begin{proof}
We need to show that for every open $U$, $X_1^{-1}(U)=_\P X_2^{-1}(U)$, which holds if $\P(X_1^{-1}(U)\cap X_2^{-1}(U))=\P(X_1^{-1}(U))=\P(X_2^{-1}(U))$
Computing a $\P$-lower-Cauchy sequence representing $W_i=X_i^{-1}(U)$ yields open sets $W_{i,n}$ such that $\P(W_{i,n}\cap W_i)\geq \P(W_i)-2^{-n}$.
Then $\Pr(X_1\in U\wedge X_2\in U)+2^{-n}=\P(W_1\cap W_2)+2^{-n}\geq \P(W_{1,n}\cap W_2)=\Pr(I\in W_{1,n}\wedge X_2\in U)=\Pr(Y_2\in W_{1,n}\times U)=\Pr(Y_1\in W_{1,n}\times U)=\Pr(I\in W_{1,n}\wedge X_1\in U)=\P(W_{1,n}\cap W_1)=\P(W_1)=\Pr(X_1\in U)$.
Since $n$ is arbitrary, $\Pr(X_1\in U\wedge X_2\in U)\geq\Pr(X_1\in U)$.
The result follows by symmetry.
\end{proof}

\section{Random Variables in Metric Spaces}

In this section, we consider random variables in metric spaces.
We show that an equivalent notion to our general random variables is given by completion in the Fan metric.

\subsection{Constructions in metric spaces}

We first prove some generally-useful results on constructions of topological partitions in computable metric spaces.

The following decomposition result is essentially a special case of the effective Baire category theorem~\cite{YasugiMoriTsujii1999,Brattka2001}.
\begin{lemma}
\label{lem:topologicalpartition}
Let $X$ be an effectively separable computable metric space, and $\nu$ be a valuation on $X$.
Then given any $\epsilon>0$, we can compute a topological partition $\calB$ of $X$ such that $\diam(B)\leq\epsilon$ for all $B\in\calB$, and $\nu(X\setminus\bigcup\calB)=0$.
\end{lemma}
\begin{proof}
For any $\epsilon>0$, any $\delta>0$, and any $x\in\tpX$, $\{ r >0 \mid \epsilon/2 < r < \epsilon \,\wedge\, \nu(\clB(x,r)\setminus B(x,r))<\delta\}$ is a computable open dense subset of $[\epsilon/2,\epsilon]$.
We can therefore construct a sequence of rationals $q_k\in(\epsilon/2,\epsilon)$ such that $|q_k-q_{k+1}|<2^{-k-1}$ and $\nu(\clB(x,q_k)\setminus B(x,q_k))<2^{-k}$.
Then taking $r_\epsilon(x)=\lim_{k\to\infty} q_k$ yields a radius $r_\epsilon(x) \in [\epsilon/2,\epsilon]$ such that $\nu(B(x,r_\epsilon(x))=\nu(\clB(x,r_\epsilon(x)))$.

Since $X$ is effectively separable, it has a computable dense sequence $(x_n)_{n\in\N}$.
For $\epsilon>0$, the sets $B(x_n,r_{\epsilon}(x_n))$ have radius at least $\epsilon/2$, so cover $X$.
We take as topological partition the sets $B(x_n,r_{\epsilon}(x_n)) \setminus \bigcup_{m=0}^{n-1} \clB(x_m,r_{\epsilon}(x_m))$ for $n\in\N$.
\end{proof}

\begin{lemma}
\label{lem:almostsureretraction}
Let $\tpX$ be an effectively separable computable metric space, and $\nu$ a valuation on $\tpX$.
Then for any $\epsilon>0$, we can construct an open set $U:\tpOp(\tpX)$ and a function $r:U\to X$ such that $\nu(U)=1$, $r$ has finite range, and $d(r(x),x)<\epsilon$ for all $x\in U$.
\end{lemma}
\begin{proof}
Let $\calB$ be the topological partition computed from Lemma~\ref{lem:topologicalpartition}.
For every $B\in\calB$, compute an element $x_B$ of $\tpX$ in $B$.
Define $r(x)$ by for $x\in\bigcup\calB$ by $r(x)=x_B$ for $x\in B$.
\end{proof}

\subsection{The Fan metric}

Let $\tpX$ be a computable metric space.
For a closed set $A$, define $d(x,A)=\inf\{d(x,y)\mid y\in A\}$ and $\clN_\epsilon(A):=\{x\in \tpX \mid d(x,A)\leq \varepsilon\}$
For an open set $U$ define $I_\varepsilon(U):=\tpX\setminus(\clN_\varepsilon(\tpX\setminus U))=\{x\in U \mid \exists \delta>0, B(x,\varepsilon+\delta)\subset U\}$.
Since $d(x,A)$ is computable in $R^{+,\infty}_<$ by our definition of a computable metric space, $\clN_\epsilon(A)$ is computable as a closed set, so $I_\varepsilon(U)$ is computable as an open set.
Note that $I_{\varepsilon_1+\varepsilon_2}(U) \subset I_{\varepsilon_1}(I_{\varepsilon_2}(U))$.

If $(\tpX,d)$ is a complete metric space, then the \emph{Fan metric} is a natural distance function on random variables:
\begin{definition}[Fan metric]
\begin{equation} \label{eq:randomvariablefanmetric}
 \begin{aligned} d_F(X,Y) &= \sup\!\big\{ \varepsilon\in\Q^+ \mid \ \P\big(\{\omega\in\Omega \mid d(X(\omega),Y(\omega))>\varepsilon\}\big) > \varepsilon\big\} \\
   &= \inf\!\big\{ \varepsilon \in \Q^+ \mid \ \P\big(\{\omega\in\Omega \mid d(X(\omega),Y(\omega))\geq\varepsilon\}\big) < \varepsilon\big\} . \end{aligned}
\end{equation}
\end{definition}
Given a computable metric $d:\tpX\times \tpX\fto\R^+$, the Fan metric on continuous random variables is easily seen to be computable:
The convergence relation defined by the Fan metric corresponds to \emph{convergence in probability}: A sequence of random variables $X_n$ taking values in a metric space converges in probability to a random variable $X_\infty$ if $\P(d(X_n,X_\infty)>\epsilon)\to0$ for all $\epsilon>0$.
If the metric $d$ on $\tpX$ is bounded, the distance \( \textstyle d(X,Y) := \int_\Omega d(X(\omega),Y(\omega)) \, dP(\omega) . \) is equivalent to the Fan metric.

Recall that for topological spaces, a sequence of random variables $(X_n)$ converges to $X_\infty$ if all open $U$, $\liminf_{n\to\infty}\P(X_n\in U\wedge X_\infty\in U)\geq \P(X_\infty\in U)$, which corresponds to $\P(X_n\in U\wedge X_\infty\in U)\to \P(X_\infty\in U)$ with convergence of probabilities being considered in $\R^+_<$.

\begin{proposition}
\label{prop:fanmetriccomputable}
The Fan metric is computable.
\end{proposition}
\begin{proof}
The random variable $X\times Y$ is computable given $X,Y$.
Then $\P(\{\omega\mid d(X(\omega),Y(\omega))>\varepsilon)=\Pr\bigl(X\times Y\in d^{-1}(\{e\mid e>\varepsilon\})\bigr)$ is computable in $\R^+_<$, and $\P\bigl(\{\omega\mid d(X(\omega),Y(\omega))>\varepsilon\bigr)>\varepsilon$ is verifiable, so $d_F(X,Y)$ is computable in $\R^+_<$.
Similarly, $\P(\{\omega\mid d(X(\omega),Y(\omega))\geq\varepsilon)$ is computable in $\R^+_>$, so $d_F(X,Y)$ is computable in $\R^+_>$.
\end{proof}

\begin{theorem}
\label{thm:fanmetriclimit}
Suppose $(X_n)_{n\in\N}$ is a fast Cauchy sequence of random variables in the Fan metric. Then $X_\infty=\lim_{n\to\infty}X_n$ a random variable which is computable from the $X_n$.
\end{theorem}
The proof is based on the non-effective version of this result from~\cite{MannWald1943}.
\begin{proof}
Given open $U$, we need to compute $X_\infty^{-1}(U)$.
Let $U_{k+1}=I_{2^{-k}}(U)$.
Then for $m>n$, Now $x_n\in U_n$ and $d(x_n,x_m)<2^{-n}$, then $x_m\in U_{m}$ since $I_{2^{-(n-1)}}(U)\subset I_{2^{-n}}(I_{2^{-(m-1)}}(U))$, and since $\Pr(d(X_m,X_n)>2^{-n})\leq 2^{-n}$, we have $\Pr(X_m\in U_{m}\wedge X_n\in U_{n})\geq \Pr(X_n\in U_{n})-2^{-n}$.
Let $W_n(U)=X_n^{-1}(I_{2^{-(n-1)}}(U))$, so $\P(W_m\cap W_n)\geq \P(W_n)-2^{-n}$.
Then $W_n(U)$ is a fast lower-Cauchy sequence of $\P$-lower-measurable sets, so converges effectively to some $\P$-lower-measurable set $W_\infty(U)$.
We define $X_\infty^{-1}(U)=W_\infty(U)$, and note that $\P(W_n\cap W_\infty)\geq \P(W_n)-2^{-n}$ for all $n$.

It is straightforward to check $X_\infty^{-1}(U)=\lim_{n\to\infty}X_n^{-1}(U)$ in the class of $\P$-lower-measurable sets.
The modularity property of $X_\infty^{-1}$ follows from the equations $\P(X_n^{-1}(U_1\cup U_2))+\P(X_n^{-1}(U_1\cap U_2))=\P(X_n^{-1}(U_1))+\P(X_n^{-1}(U_2))$ by passing through the limit.
\end{proof}
A random variable can therefore be represented by a sequence $(X_0,X_1,X_2,\ldots)$ of random variables from some simpler class satisfying $d(X_m,X_n)  < 2^{-\min(m,n)}$, and two such sequences are equivalent (represent the same random variable) if $d(X_{1,n},X_{2,n}) \to 0$ as $n\to\infty$.

\subsection{Representation}

In this section, we prove two representation results on random variables in metric spaces.
We show that we can \emph{construct} a random variable with a given \emph{distribution}, and that given any random variable, we can construct a sequence of simple continuous random variables converging effectively to it.
These results depends on the base space $\Omega=\{0,1\}^\omega$ being totally disconnected.
Clearly, for a connected base space $\Omega$, such as the interval $[0,1]$, then any continuous random variable takes values in a single component of $\tpX$, and any simple continuous random variable is constant, but if $\tpX$ is contractible, then the continuous random variables may still be dense.

The following result shows that random variables can be represented by a sequence of simple random variables.
It is a variant of~\cite[Theorem~14]{SchroederSimpson2006JCMPLX} and~\cite[Theorem~1.1.1]{HoyrupRojas2009INFCOMP}, which shows that any distribution is effectively measurably isomorphic to a distribution on $\{0,1\}^{\omega}$, and the proof is similar.
\begin{theorem}
\label{thm:randomvariabledistribution}
Let $\tpX$ be a computable metric space, and $\nu$ be a valuation on $\tpX$.
Then we can construct a random variable $X$ on base space $\Omega=\{0,1\}^\omega$ such that for any open $U$, $\Pr(X\in U)=\nu(U)$.
Further, $X$ can be constructed as the effective limit of a fast Cauchy sequence of simple continuous random variables.
\end{theorem}
\begin{proof}
For each $n$, use Lemma~\ref{lem:topologicalpartition} to construct a countable topological partition $\calB_n$ such that each $B\in B_n$ has diameter at most $2^{-n}$, and $\sum_{B\in\calB_n}\nu(\partial B)=0$.
By taking intersections if necessary, we can assume that each $\calB_{n+1}$ is a refinement of $\calB_n$.

We now construct random variables $X_n$ as follows.
Suppose we have constructed cylinder sets $W_{n,m}\subset\{0,1\}^\infty$ such that $P(W_{n,m})<\nu(B_{n,m})$ and $\sum_{m} P(W_{n,m})>1-2^{-n}$.
Since $B_{n,m}$ is a union of open sets \linebreak[4] $\{B_{n+1,m,1},\ldots,B_{n+1,m,k}\}\subset\calB_{n+1}$, we can effectively compute dyadic numbers $p_{n+1,m}$ such that $p_{n+1,m,k}<\mu(B_{n+1,m,k})$ and $\sum_k p_{n+1,m,k}\geq P(W_{n,m})$.
We then partition $W_{n,m}$ into cylinder sets $W_{n+1,m,k}$ each of measure $p_{n,m,k}$.
For each $n,m,k$ we construct a point $x_{n+1,m,k}\in B_{n+1,m,k}$, and take $X_{n+1}$ to map $W_{n+1,m,k}$ to a point $x_{n+1,m,k}\in B_{n+1,m,k}$.
It is clear that $X_{n}$ is a strongly-convergent Cauchy sequence, so is a representation of a measurable random variable $X_\infty$.

It remains to show that $\Pr(X_\infty \in U) = \nu(U)$ for all $U\in\opset(X)$.
This follows since for given $n$ we have $\Pr(X_n \in U) > \nu(I_{2^{1-n}}(U))-2^{-n} \nearrow \nu(U)$ as $n\to\infty$.
\end{proof}



\begin{theorem}
\label{thm:fastcauchysimplepiecewisecontinuous}
Let $\tpX$ be a computable metric space and $X:\tpRV(\tpX)$ a random variable.
Then one can construct a fast Cauchy sequence of simple continuous random variables $X_n$ such that $\lim_{n\to\infty}X_n=X_\infty$.
\end{theorem}
Although one could prove this directly, there is a simple proof based on Theorem~\ref{thm:randomvariabledistribution} and Proposition~\ref{prop:equalproductdistribution}:
\begin{proof}
Let $I$ be the identity random variable on $\Omega$, and let $Y=I\times X:\Omega\mfto\Omega\times\tpX$.
Let $Y_n=I_n\times X_n$ be a sequence of simple random variables with limit $Y_\infty=I_\infty\times X_\infty$ such that $\Pr[Y_\infty]=\Pr[Y]$.
Then $Y_\infty=Y$, so $X_\infty=\lim_{n\to\infty}X_n$ as required.
\end{proof}

\begin{corollary}
If $\tpX$ is a computable metric space, then the representation of random variable $X$ by its preimage $X^{-1}:\tpOp(\tpX)\to \tpMe_<(\Omega,\P)$ is equivalent to the representation by fast Cauchy sequences of simple (piecewise-)continuous random variables.
\end{corollary}
\begin{proof}
Given a representation of $X$, we can compute a fast Cauchy sequence of (piecewise-)continuous random variables by Theorem~\ref{thm:fastcauchysimplepiecewisecontinuous}. Conversely, given a fast Cauchy sequence of piecewise-continuous random variables $X_n$, we can compute $X_n$ as a measurable random variable by Observation~\ref{obs:piecewisecontinuous}, and the limit by Theorem~\ref{thm:fanmetriclimit}.
\end{proof}

\subsection{Expectation}

For bounded random variables taking values in the reals, the expectation is defined in the usual way:
\begin{definition}[Expectation]
If $X:\Omega\mfto\tpRe$ is an effectively bounded real-valued random variable, the \emph{expectation} of $X$ is given by the integral
\[ \textstyle \Ex(X) = \int_{\tpRe} x \, d\Pr[X] , \]
where $\Pr[X]$ is the valuation on $\tpRe$ induced by $X$ i.e. $\Pr[X](U)=\Pr(X\in U)=\P(X^{-1}(U))$, and the integral is given by Definition~\ref{defn:boundedintegral}.
\end{definition}

The expectation of possibly unbounded real-valued random variables is not continuous in the weak topology; for example, we can define continuous random variables $X_n$ taking value $2^n$ on a subset of $\Omega$ of measure $2^{-n}$, so that $X_n\to0$ but $\Ex(X_n)=1$ for all $n$.
For this reason, we need a new type of \emph{integrable random variables}.
\begin{definition}[Integrable random variable]
Let $(\tpX,d)$ be metric space with distinguished element $z$ (e.g. the zero of a normed space), and let $Z$ be the constant random variable $Z(\omega)=z$.
Let $d_1$ be the (possibly infinite-valued) distance function
\begin{equation} \label{eq:randomvariableuniformmetric} d_1(X,Y) = \int_\Omega d(X(\omega),Y(\omega)) dP(\omega) . \end{equation}
The type of \emph{integrable random variables} $L^1(\tpX)$ is the completion of the set of all effectively bounded random variables $X$ such that $d_1(X,Z)<\infty$.
Then $d_1$ is a metric on $L^1(\tpX)$.
\end{definition}
\noindent
If $d$ is a bounded metric, then this metric is equivalent to the Fan metric.

For continuous and integrable real-valued random variables, then the expectation is also given by an integral over the base space $\Omega$.
\begin{proposition}[Expectation]\mbox{}
\begin{enumerate}[(i)]
\item If $X:\Omega \fto \tpRe$ is a continuous real-valued random variable, then the expectation of $X$ is given by the integral
\[ \textstyle \Ex(X) = \int_{\Omega} X(\omega) \, d\P(\omega) , \]
which always exists since $X$ has compact values.
\item If $X:\Omega \mfto \tpRe$ is an integrable real-valued random variable, and $X$ is presented as $\lim_{n\to\infty}X_n$ for some sequence of continuous random variables satisfying $\Ex[|X_{n_1}-X_{n_2}|]\leq 2^{-\min(n_1,n_2)}$, then $\bigl(\Ex(X_{n})\bigr)_{n\in\N}$ is an effective Cauchy sequence, and
\[ \textstyle \Ex(X) = \lim_{n\to\infty} \Ex(X_n). \]
\end{enumerate}
\end{proposition}
\begin{proof}\mbox{}
\begin{enumerate}[(i)]
\item It suffices to consider the case $X\geq0$.
Then both $\int_\R x d\Pr[X]$ and $\int_\Omega X(\omega)d\P(\omega)$ yield Choquet integral sums of the form $\sum_{m=1}^{n}(p_m-p_{m-1})\P(X^{-1}(p_m,\infty))$.
\item The expectation is continuous. \qedhere
\end{enumerate}
\end{proof}

We can effectivise Lesbegue spaces $\mathcal{L}^p$ of integrable random variables through the use of effective Cauchy sequences in the natural way:
If $(\tpX,|\cdot|)$ is a normed space, then the type of $p$-integrable random variables with values in $\tpX$ is the effective completion of the type of $p$-integrable continuous random variables under the metric $d_{p}(X,Y)=||X-Y||_p$ induced by the norm
\begin{equation} \label{eq:randomvariableuniformnorm} \textstyle ||X||_{p} = \left( \strut \int_{\Omega} | X(\omega)|^p\,dP(\omega) \right)^{1/p} = \bigl(\Ex(|X|^p)\bigr)^{1/p}. \end{equation}
We can easily prove the Cauchy-Schwarz and triangle inequalities for measurable random variables
\( ||XY||_{pq/(p+q)} \leq ||X||_{p}  \cdot ||Y||_{q} \quad \text{and} \quad ||X+Y||_{p} \leq ||X||_{p}+||Y||_{p} \,. \)

%

%


The following result relates the expectation of a random variable to an integration over its valuation.
An analogous result in a different setting~\cite[Theorem~15]{SchroederSimpson2006JCMPLX}.
\begin{theorem}[Expectation]
\label{thm:expectation}
Let $X$ be a positive real-valued random variable such that $\Ex(X)<\infty$.
Then
\begin{equation*} \textstyle \Ex(X) = \int_{0}^{\infty} \Pr(X>x) dx = \int_{0}^{\infty} \Pr(X \geq  x) dx . \end{equation*}
\end{theorem}
\noindent
Note that the first integral is computable in $\tpR^+_<$, but the second integral is in general uncomputable in $\tpR^+_>$, due to the need to take the limit as the upper bound of the integral goes to infinity.
However, the second integral may be computable if the tail is bounded, for example, if $X$ takes bounded values.
The proof follows from the definition of the lower integral:
\begin{proof}
First assume $X$ is a continuous random variable, so by definition, $\Ex(X) = \int_\Omega X(\omega)\,dP(\omega)$.

The definition of the lower horizontal integral gives
$\int_\Omega X(\omega)\,dP(\omega) \geq \sum_{i=0}^{n-1} (x_i-x_{i-1}) P(\{\omega\mid X(\omega)>x_i\})$
for all values $0=x_0<x_1<\cdots<x_n$.
Take $x_{i}-x_{i-1}<\epsilon$ for all $i$.
Then $\Ex(X)+\epsilon = \int_\Omega X(\omega)+\epsilon\,dP(\omega)  \geq \sum_{i=1}^{n} (x_i-x_{i-1}) P(\{\omega\mid X(\omega) +\epsilon > x_i\}) = \sum_{i=1}^{n} \int_{x_{i-1}}^{x_i} \Pr(X(\omega) > x_i-\epsilon) \,dx  \geq \sum_{i=1}^{n} \int_{x_{i-1}}^{x_i} \Pr(X(\omega) > x_{i-1}) \,dx \geq \sum_{i=1}^{n} \int_{x_{i-1}}^{x_i} \Pr(X(\omega) > x) \, dx = \int_{0}^{x_{n-1}} \Pr(X(\omega) > x)$.
Taking $n\to\infty$ gives $\Ex(X) \geq \int_{0}^{\infty}\Pr(X>x)dx - \epsilon$, and since $\epsilon$ is arbitrary, $\Ex(X) \geq \int_{0}^{\infty} \Pr(X > x) dx$.

The definition of the lower horizontal integral gives for all $\epsilon>0$, there exist $0=x_0<x_1<\cdots<x_n$, such that $\int_\Omega X(\omega)\,dP(\omega) \leq \sum_{i=1}^{n} (x_i-x_{i-1}) P(\{\omega\mid X(\omega)>x_i\})+\epsilon$.
By refining the partition if necessary, we can assume $x_i-x_{i-1}<\epsilon$ for all $i$.
Then $\Ex(X)-\epsilon \leq \sum_{i=1}^{n} (x_i-x_{i-1}) P(\{\omega\mid X(\omega)>x_{i}\}) = \sum_{i=1}^{n} \int_{x_{i-1}}^{x_{i}} P(\{\omega\mid X(\omega)>x_{i}\} dx \leq \sum_{i=1}^{n} \int_{x_{i-1}}^{x_{i}} P(\{\omega\mid X(\omega)>x\} dx=\int_{0}^{x_n} P(\{\omega\mid X(\omega)>x\} dx\leq\int_{0}^{\infty} P(\{\omega\mid X(\omega)>x\} dx$.
Hence $\Ex(X) \leq \int_{0}^{\infty}\Pr(X>x)dx + \epsilon$, and since $\epsilon$ is arbitrary, $\Ex(X) \leq \int_{0}^{\infty} \Pr(X > x) dx$.

The case of measurable random variables follows by taking limits.

We show $\int_{0}^{\infty}\Pr(X\geq x)dx = \Ex(X)$ since  $\int_{0}^{\infty}\Pr(X > x)dx \leq \int_{0}^{\infty}\Pr(X\geq x)dx \leq \int_{0}^{\infty}\Pr(X + \epsilon > x)dx = \epsilon + \int_{0}^{\infty}\Pr(X  \geq x)dx$  for any $\epsilon>0$.
\end{proof}
\noindent
By changing variables in the integral, we obtain:
\begin{corollary}
If $X$ is a real-valued random variable, then for any $\alpha\geq 1$,
\[ \textstyle \Ex(|X|^\alpha) = \int_{0}^{\infty} \alpha \, x^{\alpha-1} \, \Pr(X>x) dx = \int_{0}^{\infty} \alpha x^{\alpha-1} \Pr(X \geq  x) dx . \]
\end{corollary}

\begin{remark}[Expectation of a distribution]
Theorem~\ref{thm:expectation} shows that the expectation of a random variable depends only on its \emph{distribution}.
Indeed, we can define the expectation of a probability \emph{valuation} $\pi$ on $[0,\infty[$ by
\[ \textstyle \Ex(\pi) = \int_{0}^{\infty} \pi(\,]x,\infty[\,) dx  = \int_{0}^{\infty} \pi(\,[x,\infty[\,) dx  . \]
\end{remark}
\noindent
If $f:\tpX\fto\tpRe^+_<$, then we can compute the \emph{lower expectation} of $f(X)$ by
\begin{equation}\label{eq:lowerexpectation} \textstyle \Ex_<(f(X)) := \int_{0}^{\infty} \Pr\bigl(X\in f^{-1}(\,]\lambda,\infty[\,)\bigr) d\lambda . \end{equation}


We have an effective version of the classical dominated convergence theorem.
\begin{theorem}[Dominated convergence]
Suppose $X_n\to X$ weakly, and there is an integrable function $Y:\Omega \fto \tpRe$ such that $|X_n| \leq Y$ for all $n$ (i.e. $\Pr(Y-|X_n|\geq0)=1$) and that $\Ex|Y|<\infty$.
Then $X_n$ converges effectively under the metric~\eqref{eq:randomvariableuniformmetric}.
In particular, the limit of $\Ex(X_n)$ always exists
\end{theorem}
\begin{proof}
Since $\Ex(Y) < \infty$, the probabilities $\Pr(Y \geq y) \to 0$ as $y\to\infty$.
For fixed $\epsilon>0$, let $b(\epsilon)=\sup\{ y\mid \Pr(Y\geq y) \geq \epsilon$, which is computable in $\R_>$ given $\epsilon$.
Then $\sup\{ \int_A Y dP \mid P(A)\leq \epsilon \} \leq \int_{b(\epsilon)}^{\infty} \Pr(Y\geq y) dy=\Ex(Y)-\int_{0}^{b(\epsilon)} \Pr(Y > y) dy$ in $\R^+_>$.
For continuous random variables $X_m$, $X_n$ with $2^{-m},2^{-n}<\epsilon$, taking $A_\epsilon=\{\omega\mid d(X_m(\omega),X_n(\omega))\geq \epsilon\}$ gives
\( \Ex(|X_m-X_n|) \leq \epsilon + \int_{A_\epsilon}|X_m(\omega)-X_n(\omega)|\,dP(\omega) \leq \epsilon + \int_{A_\epsilon}|X_m(\omega)|+|X_n(\omega)|dP(\omega)  \leq \epsilon + \int_{A_\epsilon} 2|Y| dP \leq \epsilon + 2 \int_{b(\epsilon)}^{\infty} \Pr(Y\geq y) dy, \)
which converges effectively to $0$ as $\epsilon\to 0$.
\end{proof}


\section{Conditioning}
\label{sec:conditioning}

The concept of conditional random variable is subtle even in classical probability theory.
The basic idea is that if we condition a random quantity $Y$ on some information of kind $\calX$, then we can reconstruct $Y$ given a value $x$ describable by $\calX$.
Classically, conditional \emph{random variables} are \emph{not} defined, but conditional distributions and expectations are.
Conditional expectations can be shown to exist using the Radon-Nikodym derivative, but this is uncomputable~\cite{HoyrupRojasWeihrauch2011}.

\subsection{Independence}
\label{sec:independence}

In the classical case, we condition relative to a sub-sigma-algebra of the measure space.
In the computable case, it makes sense to consider instead a sub-topology $\calT$ on $\Omega$.
We first need to define concepts of $\calT$ measurability and $\calT$ independence
\begin{definition}[Measure-topologies]
Let $\nu$ be a valuation on $\tpX$.
A \emph{$\nu$-topology} is a collection of $\nu$-lower-measurable sets which contains $\emptyset,\tpX$ and is closed under intersection and countable union.
\par
The $P$-topology \emph{generated} by a random variable $X:\Omega\mfto\tpX$ is simply $\{ X^{-1}(U) \mid U\in\opset(\tpX)\}$.
A random variable $X$ is $\calT$-measurable if $X^{-1}(U)\in\calT$ for all $U\in\tpOp(\tpX)$
\par
We write $\tpRV_\calT(\tpX)$ for the type of $\calT$-measurable random variables with values in $\tpX$.
\end{definition}
Note that a $\nu$-topology is not a topology on $\tpX$ in the standard sense, since it consists of equivalence-classes of subsets of $\tpX$, rather than sets themselves.

Recall that classically, we say random variables $X_1,X_2$ taking values in $\tpX_1$, $\tpX_2$ are independent if for all open $U_1\subset \tpX_1$ and $U_2\subset \tpX_2$, we have \( \Pr(X_1\in U_1 \wedge X_2\in U_2) = \Pr(X_1\in U_1) \cdot \Pr(X_2\in U_2). \)
This classical definition does not relate well with computability theory, as the following example shows:
\begin{example}
Consider the result $X$ of throwing a $6$-sided die, and the random variables $X_{\mathrm{even}}$ which is $1$ if $X$ is even and $0$ otherwise, and $X_{\mathrm{high}}$ which is $1$ is $X$ is a~5 or~6.
Then $X_\mathrm{even}$ and $X_\mathrm{high}$ are independent for a fair die, but not if the probability of a~6 is $\tfrac{1}{6}+5\epsilon$ and of a~1~to~5 is $\tfrac{1}{6}-\epsilon$ for $\epsilon\neq0$.
\end{example}

It is therefore useful to consider different versions of independence  properties.
\begin{definition}[Independence]
\label{defn:independentrandomvariables}
\mbox{}\par
$P$-topologies $\calT_{1,2}$ on $\Omega$ are \emph{independent} if $P(U_1\cap U_2)=P(U_1)P(U_2)$ for all $U_1\in\calT_1$, $U_2\in\calT_2$.
\par
$P$-topologies $\calT_{1,2}$ are \emph{strongly independent} if we can write $\Omega=\Omega_1\times\Omega_2$, with projections $p_{1,2}$ and inclusion $q:\Omega_1\times\Omega_2\to\Omega$ such that $q(p_1(U_1),\omega_2)=U_1$ for all $U_i\in\calT_i$, there exists $V_i\subset\Omega_i$ such that $U_i=p_i^{-1}(V_i)$ for $i=1,2$.
\par
Random variables $X_{1,2}:\Omega\mfto\tpX_{1,2}$ are \emph{effectively independent} if $X_1:\tpRV_{\calT_1}(\tpX_1)$ and $X_2:\tpRV_{\calT_2}(\tpX_2)$ for independent topologies $\calT_{1,2}$.
\par
Random variables $X_1,\ldots,X_k$ are \emph{jointly} independent of $\calT$ if the product $\prod_{i=1}^{k}X_i$ is independent of $\calT$, and $X_1,X_2,\ldots$ are jointly independent of $\calT$ if every finite product is independent of $\calT$.
\par
A random variable $X$ is \emph{effectively independent} of a topology $\calT$ on $\Omega$ if $X:\tpRV_{\calX}(\tpX_1)$ for some topology $\calX$ independent of $\calT$.
\end{definition}
\noindent
We can express independence relative to sub-topologies using the identity random variable~$I:\Omega\fto\Omega$.
If $X:\Omega\mfto\tpX$ is independent of $\calT$, then
\( \Pr(X\times I \in U \times W) = \Pr(X \in U) P(W) \)
whenever $U\in\opset(\tpX)$ and $W\in\calT$.
We write $\tpRV_{\perp\calT}(\tpX)$ for $\calT$-independent random variable with values in $\tpX$.
Note that $\tpRV_{\perp\calT}(\tpX)$ does not form a natural type, since it is possible for $X_1$, $X_2$ to be independent of $\calT$, but $X_1\times X_2$ not to be.
If $X$ is effectively independent of $\calT$, and $Y$ is $\calT$-measurable, then $X$ is effectively independent of $Y$.

\subsection{Conditional Random Variables}
\label{sec:conditionalrandomvariable}

We now proceed to our notion of conditional random variable $Y|\calX$, where $\calX$ is a measure-topology on $\Omega$.
Recall that classically, the conditional expectation $\Ex(Y|\calX)$ for a random variable $Y$ is a $\calX$-measurable random variable.
\begin{definition}[Conditional random variable]
\label{defn:conditionalrandomvariable}
Let $\calX$ be a measure-topology. A $\calX$-independent conditional random variable is a function $Y|:\tpX\to\tpRV(\tpY)$ such that the $Y|x$ are jointly-independent of $\calX$.
\end{definition}

If $Y|:\tpX\to\tpRV(\tpY)$ is $\calX$-independent, and $X:\tpRV(\tpX)$ is simple and $\calX$-measurable, then we can define the \emph{joint random variable}
\begin{definition}[Joint random variable]
\label{defn:jointrandomvariable}
The \emph{joint random variable} of $X \rtimes Y|$ of a $\calX$-measurable simple random variable $X:\tpRV(\tpX)$ and a $\calX$-independent conditional random variable $Y|:\tpX\to\tpRV(\tpY)$ is defined by
\begin{equation} \label{eq:jointrandomvariable} \textstyle (X \rtimes Y|)^{-1}(U\times V) = \bigcup_{x_i\in U} X^{-1}(x_i) \cap (Y|x_i)^{-1}(V) . \end{equation}
The \emph{joint random variable} of a $\calX$-measurable random variable $X:\tpRV(\tpX)$ and a $\calX$-independent conditional random variable \nolinebreak[4]{$Y|:\tpX\to\tpRV(\tpY)$} is defined to be $\lim_{n\to\infty} X_n \rtimes Y|$, where $X_n$ is a sequence of $\calX$-measurable simple random variables converging to $\calX$.
\end{definition}
\noindent
Note that if $X$ is a simple random variable,
\[ \textstyle \P\bigl(\bigcup_{x_i\in U} X^{-1}(x_i) \cap (Y|x_i)^{-1}(V)\bigr)=\sum_{x_i\in U} \Pr(X=x_i) \times \Pr(Y|x_i\in V) . \]

To show Definition~\ref{defn:jointrandomvariable} makes sense in the general case, we need to show that it is independent of the sequence of simple random variables used to specify $X$.
\begin{lemma}
Let $X$ be a $\calX$-measurable random variable in a metric space $\tpX$.
Then we can construct a sequence of $\calX$-measurable simple continuous random variables converging effectively to $X$.
\end{lemma}
\begin{proof}
Let $r_n:U_n\to \tpX$ be a finite-valued map defined on an open set $U_n$ with $\Pr(X\in U_n)=1$ such that $d(r_n(x),x)<2^{-n}$ for all $n$, as guaranteed by Lemma~\ref{lem:almostsureretraction}.
Take $X_n=r_{n+1}\circ X$, which is $\calX$-measurable since $X$ is.
Then $(X_n)$ is sequence of random variables with $d(X_n,X)\leq 2^{-(n+1)}$, and  $d(X_{n_1},X_{n_2})\leq 2^{-\min(n_1,n_2)}$ as required.
\end{proof}

\begin{theorem}\label{thm:conditionalrandomvariable}
The joint random variable of a $\calX$-measurable random variable $X:\tpRV(\tpX)$ and a $\calX$-independent conditional random variable $Y|:\tpX\to\tpRV(\tpY)$ is independent of the sequence of simple approximations $X_n$ to $X$ used in the definition, so is computable.
\end{theorem}
\begin{proof}
Define the continuity sets
\[ C_{\delta,\epsilon} = \{ x\in \tpX \mid \forall \tilde{x}\in\tpX,\ d(x,\tilde{x})\leq\delta \implies d_F(Y|x,Y|\tilde{x})<\epsilon \} . \]
Note that every $C_{\delta,\epsilon}$ is open, and for any fixed $\epsilon$, continuity of $Y|:\tpX\fto\tpRV(\tpY)$ implies $\bigcup_{\delta>0} C_{\delta,\epsilon} = \tpX$.
Hence for fixed $\epsilon$ and some $\delta<\epsilon/8$ sufficiently small, $\Pr(X\in C_{\delta,\epsilon/4})>1-\epsilon/4$.

Now suppose that $X_1,X_2$ are simple $\calX$-measurable random variables such that $\Pr(d(X,X_i)<\delta)>1-\epsilon/8$.
Let \[ U=\{(x,x_1,x_2)\mid x\in C_{\delta,\epsilon/4}\wedge d(x,x_1)<\delta \wedge d(x,x_2)<\delta\} ,\] and \[ V=\{(x_1,x_2)\mid d(Y|x_1,Y|x_2)<\epsilon/2 \wedge d(x,x_1)<\delta \wedge d(x,x_2)<\delta\} . \]
Then $(x,x_1,x_2)\in U \implies (x_1,x_2)\in V$, so
\[ \Pr(X_1\!\times\! X_2\in V)\geq \Pr(X\!\times\! X_1\!\times\! X_2\in U) \geq 1-\epsilon/4-\epsilon/8-\epsilon/8=1-\epsilon/2 .\]
Since each $Y|x$ is $\calX$-independent,
\[ \begin{aligned} &\textstyle \Pr(d(Y|X_1,Y|X_2)>\epsilon/2) \\&\qquad\textstyle= \sum_{x_1,x_2} \Pr(X_1=x_1\wedge X_2=x_2) \Pr(d(Y|x_1,Y|x_2)>\epsilon/2) \\
  &\qquad\textstyle  \leq \sum_{(x_1,x_2)\in V} \Pr(X_1=x_1\wedge X_2=x_2) \times \epsilon/2 \\&\qquad\qquad\qquad\textstyle + \sum_{(x_1,x_2)\not\in V} \Pr(X_1=x_1\wedge X_2=x_2)\times 1 \\&\qquad\qquad\textstyle  =  \Pr(X_1\!\times\!X_2\in V)\times\epsilon/2  + (1-\Pr(X_1\!\times\!X_2\in V))\times1 \\&\qquad\qquad\textstyle \leq \epsilon/2+\epsilon/2=\epsilon . \end{aligned} \]
Hence if $d_F(X,X_1),d_F(X,X_2)<\delta$, we have $d_F(Y|X_1,Y|X_2)<\epsilon$.
\end{proof}


\subsection{Random functions}

In the definition of conditional random variable, we use objects of type $\tpX\fto\tpRV(\tpY)$, which are random-variable-valued functions, rather than \emph{random functions} with type $\tpRV(\tpX\fto\tpY)$, alternatively written $\tpRV(\tpCts(\tpX;\tpY))$.

Given a random function $F:\tpRV\bigl(\tpX\to\tpY\bigr)$ and a random variable $X:\tpRV(\tpX)$, since the evaluation map $\varepsilon:(\tpX\to\tpY)\times \tpX \to \tpY$ is computable, we can apply it to $F$ and $X$ to obtain a random variable $Y=\varepsilon(F,X):\tpRV(\tpY)$.

The information provided by a random function $\tpRV(\tpX\to\tpY)$ is strictly stronger than that provided by a function $\tpX\to\tpRV(\tpY)$:
\begin{proposition}[Random function]
The natural bijection $\tpRV(\tpX\fto\tpY) \hookrightarrow (\tpX\fto\tpRV(\tpY))$ is computable, but its inverse is not continuous.
\end{proposition}
\begin{proof}
For fixed $x$, evaluation $\varepsilon_x:(\tpX\fto\tpY)\fto\tpY:f\mapsto f(x)$ is computable, so by Theorem~\ref{thm:imagerandomvariable}, $
\varepsilon(F):\tpRV(\tpY)$ is computable for any $F:\tpRV(\tpX\fto\tpY)$ given $x$. Hence the function $x\mapsto \varepsilon_x(F)$ is computable.

Conversely, let $X=\{0,1\}^\omega$ and $Y=\{0,1\}$.
Define $F(x,\omega,n)=1$ if $x|_n=\omega|_n$, and $0$ otherwise.
Then for fixed $x$, $F(d(x,\cdot,n),0)=2^{-n}$, so $F(x,\cdot,n)$ converges to $0$ uniformly in $x$.

For fixed $\omega$, $d(F(\cdot,\omega,n_1),F(\cdot,\omega,n_2)) = \sup_{x\in X} d(F(x,\omega,n_1),F(x,\omega,n_2)) = 1$, since (for $n_1<n_2$) there exists $x$ such that $x|_{n_1}=\omega|_{n_1}$ but $x|_{n_2}\neq\omega|_{n_2}$.
Hence $d(F(\cdot,\cdot,n_1),F(\cdot,\cdot,n_2))=1$ for all $n_1,n_2$, and the sequence is not a Cauchy sequence in $\tpRV(\tpX\fto\tpY)$.
\end{proof}

However, if $Y|:\tpX\to\tpRV(\tpY)$ is such that each $Y|x$ is a continuous random variable i.e. a continuous function $\Omega\to\tpY$, then $Y|$ corresponds to a continuous random function $F$ by  $[F(\omega)](x)=Y|x(\omega)$, and $Y|X$ is the random variable $\varepsilon(F,X)$.

\section{Conclusions}

In this paper, we have developed a theory of probability and random variables.
The theory uses type-two effectivity to provide an underlying machine model of computation, but is largely developed using type theory in the cartesian-closed category of quotients of countably-based spaces, which has an effective interpretation.
The approach extends existing work on probability via valuations and random variables in metric spaces via limits of Cauchy sequences.

The approach has been used to give a computable theory for stochastic processes which is sufficiently powerful to effectively compute the solution of stochastic differential equations~\cite{Collins2014ARXIV}. Ultimately, we hope that this work will form a basic for practical software tools for the rigorous computational analysis of stochastic systems.
\vspace{\baselineskip}

\noindent
\textbf{Acknowledgement:} The author would like to thank Bas Spitters for many interesting discussions on measurable functions and type theory, and pointing out the connection with monads.

\bibliographystyle{alpha}
\bibliography{computablerandomvariables}

\end{document}